\documentclass{article}
\usepackage{graphicx} 
\usepackage{epsfig, epsf,subfigure, color,longtable}
\usepackage{pstricks, pst-node, psfrag}
\usepackage{amssymb,amsmath,amsthm}
\usepackage{verbatim,enumerate,soul}
\usepackage{rotating, lscape}
\usepackage{setspace}
\usepackage{bbm}
\usepackage{natbib}
\usepackage{amsthm}
\usepackage{makecell}
\usepackage{dsfont}
\usepackage{url}
\usepackage{multirow}
\usepackage{caption}
\usepackage{amsmath}
\usepackage{algorithm}
\usepackage{algorithmic}
\usepackage{amssymb}
\usepackage{appendix}
\usepackage{tabularx,booktabs}
\usepackage{amsmath}
\usepackage{cancel}
\newtheorem{prop}{Proposition}
\newtheorem{remark}{Remark}[section]
\newtheorem{thm}{Theorem}
\newtheorem{dfn}[thm]{Definition}
\newcommand{\Rlogo}
{\protect\includegraphics[height=1.8ex,keepaspectratio]{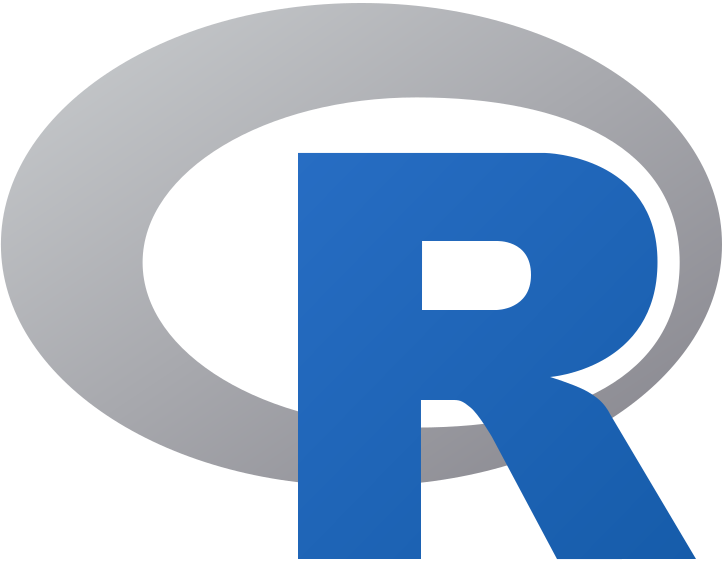}}
\graphicspath{{plots/}}
\usepackage{graphicx}
\usepackage{epstopdf}
\usepackage{tikz-cd}
\usepackage{soul}

\DeclareGraphicsExtensions{.pdf}
\newcommand{\pkg}[1]{{\normalfont\fontseries{b}\selectfont #1}}
   
\theoremstyle{plain} 
\newtheorem{theorem}{Theorem}[section]
\newtheorem{proposition}[theorem]{Proposition}

\theoremstyle{remark} 
\usepackage[shortlabels]{enumitem}

\usepackage{enumitem}
\usepackage{float}
\usepackage[colorlinks=true,citecolor=blue]{hyperref}
\usepackage{amsmath,amsfonts,amsthm,bm}
\setlist[enumerate]{wide=0pt, widest=99,leftmargin=\parindent, labelsep=*}
\definecolor{darkblue}{RGB}{0,0,139}
\definecolor{darkgreen}{RGB}{0,100,0}
\definecolor{darkviolet}{RGB}{148,0,211}

\newcommand{\commCFJV}[1]{{\leavevmode\color{green!50!black}#1}}

\begin{document}

\thispagestyle{empty} \baselineskip=28pt \vskip 5mm
	\begin{center} {\Large{Frequency-Domain Analysis of Time Series with Network-Structured Dependence: Application to Global Bank Connectedness}}
	\end{center}
	
	\baselineskip=12pt \vskip 5mm
	\begin{center}\large
		Cristian F. Jim\'enez-Var\'on\footnote[1]{
			\baselineskip=10pt Department of Mathematics. University of York. E-mail: cristian.jimenezvaron@york.ac.uk; , marina.knight@york.ac.uk }, Marina I. Knight$^1$\end{center}	
	
	\baselineskip=16pt \vskip 1mm \centerline{\today} \vskip 8mm

	\begin{center}
		{\large{\bf Abstract}}
	\end{center}
	\baselineskip=17pt
Financial spillovers in interconnected systems, such as global banking networks, require tools that capture temporal and frequency dynamics, while incorporating the underlying network topology. While current network time series models are developed in the time-domain, frequency-domain approaches, which reveal how cross-nodal dependencies vary across different cycles, remain under-explored. This paper develops a spectral analysis framework that accommodates flexible forms of network dependence, including interactions mediated through intermediate nodes. This ensures that inter-nodal relationships are not restricted to direct connections, a feature crucial for capturing indirect financial spillovers. We define the network time series spectral density, alongside coherence and partial coherence, and propose both parametric and network-constrained nonparametric methods for their estimation. Simulations and theoretical results demonstrate the strong performance of the parametric approach when the data-generating process aligns with the model structure, whereas the nonparametric alternative provides robustness against model misspecification. An application to global bank connectedness shows that the proposed spectral measures capture inter-bank frequency-specific spillover effects, yielding results consistent with existing measures while additionally uncovering richer patterns of volatility transmission that are intimately connected to the network topology.

	\begin{doublespace}
		
		\par\vfill\noindent
	\end{doublespace}
	
	\pagenumbering{arabic}

\clearpage
\section{Introduction}\label{sec:intro}

Understanding the structural interconnections among financial institutions—particularly banks—is fundamental for assessing systemic risk and ensuring financial stability. Financial systems are inherently interacting: institutions are linked through credit exposures, liquidity provision, derivative contracts, and overlapping asset holdings. These structural relationships give rise to a network topology underlying financial markets, which governs the transmission and amplification of shocks. While connectedness measures have become a central concept in financial econometrics \citep{DieboldYilmaz2009,DieboldYilmaz2014}, they can be interpreted more generally as empirical summaries of the underlying network architecture of financial interdependence \citep{Wu2020,Gabauer2023}. A central research question emerging from this literature is how specific network structures—such as direct versus multi-hop connections—govern the propagation of shocks.

To illustrate this complexity, Figure \ref{fig:banking_network} presents a global banking network consisting of 57 major financial institutions across different countries and regions, based on data from \citet{Demirer2018}. The network (bank/nodal) connectivity can be derived from physical links or prior empirical evidence, and our framework treats this structure as a known topology that shapes the evolving temporal cross-nodal interactions.

\begin{figure}[htbp]
    \centering
    \includegraphics[width=\textwidth, height=250pt, keepaspectratio]{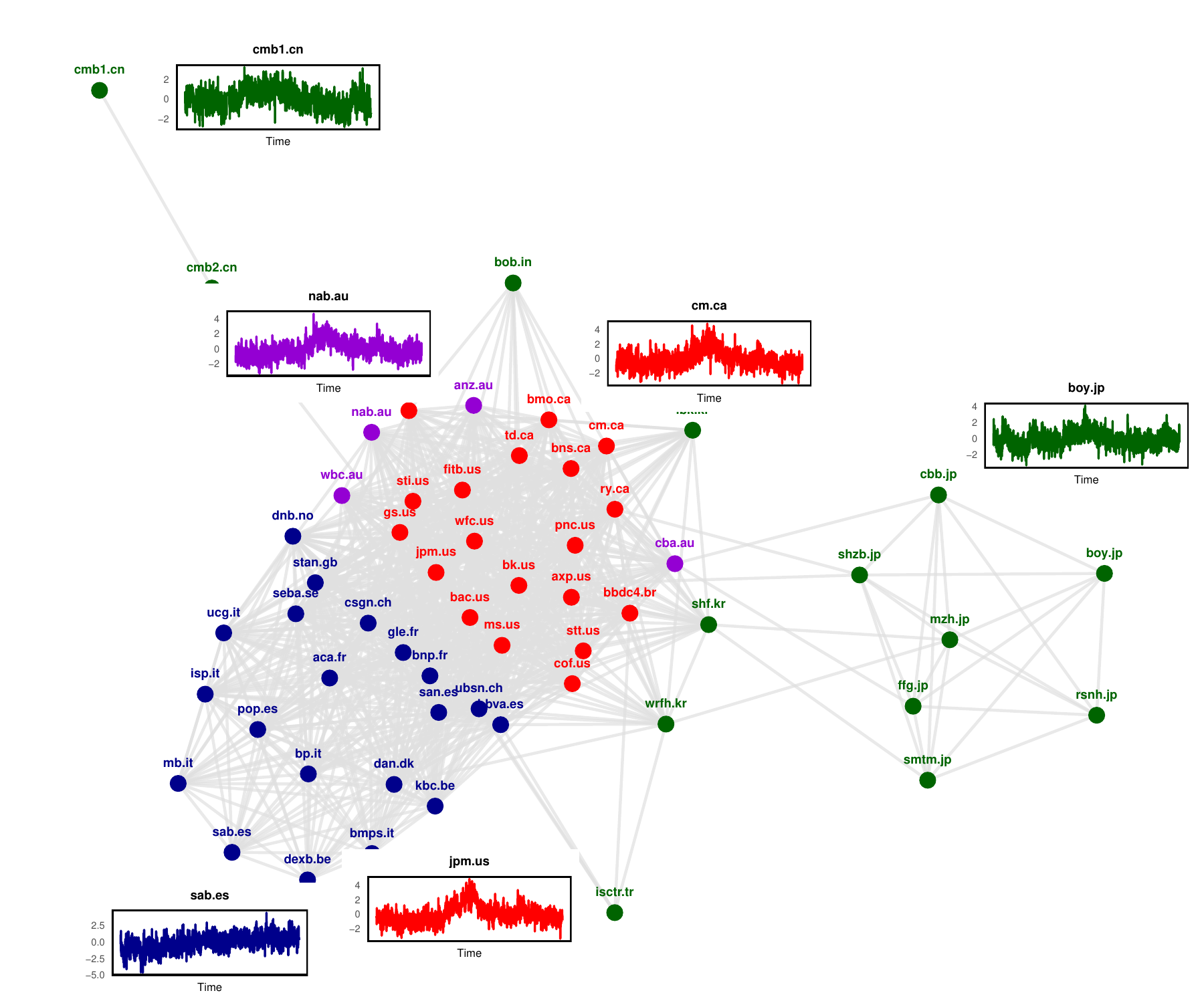}
    \caption{Global Banking Network. Each node represents a financial institution, with edges reflecting institutional links. Node colors indicate regional origin: \textcolor{red}{$\bullet$}~Americas, \textcolor{blue}{$\bullet$}~Europe, \textcolor{darkgreen}{$\bullet$}~Asia, and \textcolor{violet}{$\bullet$}~Australia. The time series shown at selected nodes track the evolution of stock return volatility, with dependencies governed by the network’s underlying topology.}
    \label{fig:banking_network}
\end{figure} 

Specifically, in this setting we do not merely observe a collection of independent time series; rather, we observe measurements collected at the nodes of a graph that evolve over time, leading to time-indexed observations constrained by a specific, known network topology. This  highlights the dual challenge of modern financial analysis: capturing the temporal dynamics of individual bank data while accounting for the institutional interdependencies that facilitate or restrict the flow of shocks across different time-spans (frequencies).

A substantial body of literature has focused on {\em modeling connectedness in the time-domain}. The framework of \citet{DieboldYilmaz2009,DieboldYilmaz2014} uses forecast error variance decompositions from vector autoregressive (VAR) models to quantify spillovers, generating a network of interdependencies. High-dimensional extensions employ regularization techniques, such as Lasso, to identify sparse structures and isolate dominant channels of risk propagation \citep{Demirer2018}. While powerful, these approaches aggregate dynamic behavior over time, masking the possibility that spillovers may operate differently over various time-horizons. 

Recognizing these limitations, \citet{Barunik2018} and \citet{Krampe2025} proposed {\em spectral-domain} frameworks to decompose variance contributions across frequency bands, identifying heterogeneous responses to shocks. However, a conceptual gap remains: these approaches treat the system as a purely data-driven multivariate or factor-based process where the network is a statistical byproduct. By contrast, they fail to explicitly incorporate a pre-defined network topology—such as the institutional links in Figure \ref{fig:banking_network}—as {\em the} structural constraint that governs the evolution of the whole system. This drives the need for a paradigm that {\em integrates graph structure} into the frequency-domain analysis, allowing for a  spillover analysis that is mediated through multi-hop connections rather than just unrestricted pairwise ones.

This challenge is not unique to finance; network-structured time series data have gained growing attention across fields such as neuroscience \citep{Bullmore2009}, where the brain's physical architecture constrains functional communication across different frequency bands \citep{Medaglia2015}. 

In the {\em time-domain} study of time series with network structured dependence (TS-NSD), substantial efforts have focused on developing multivariate autoregressive frameworks. An early approach was the construction in \cite{knight2016}, followed by the Network Vector Autoregressive (NVA) model of \cite{Zhu2017} and significantly extended by the Generalized Network Autoregressive (GNAR) framework \citep{Knight2020} where the NVA model can be seen as a particular case. Unlike earlier iterations, GNAR introduced $r$-stage neighborhood dependencies, allowing for complex interactions mediated through intermediate nodes and establishing a more flexible architecture. This framework, along with its various extensions, has proven versatile in empirical studies, including applications to international trade flows, viral spread, and wind power forecasting \citep{Nason2022,Mantziou2023}.

Within this class of models, a critical distinction arises between parameter flexibility and estimation efficiency. \cite{Yin2023} proposed a framework that relaxes stringent stationarity and Gaussianity assumptions to accommodate node-specific heterogeneity and non-Gaussian error processes. Further extensions have recently explored latent group structures \citep{Zhu2025}, edge-based processes \citep{Mantziou2023}, and smoothly time-varying linkages for locally stationary processes \citep{Chen2025, Li2025}. However, while these may provide higher flexibility, estimating distinct parameters for every node can lead to over-parameterization in large-scale networks. In contrast, frameworks that prioritize parsimony by sharing network effect parameters across neighborhoods act as a structural regularizer, enhancing estimation efficiency, numerical stability and prediction accuracy \citep{Nason2026}. 

Despite these time-domain advances, there remains a significant gap in the literature for a methodology capable of revealing the frequency-specific composition of these network relationships. Related efforts in Graph Signal Processing (GSP) have introduced spectral frameworks based on joint time-vertex stationarity \citep{Loukas2019}. These typically treat the network as a static operator—often a fixed Laplacian—used to regularize signal recovery tasks and primarily function as ``spatial'' filters rather than explicitly modeling the autoregressive pathways through which shocks propagate over time. Consequently, GSP-based methods are often restricted to immediate nodal interactions and lack the stochastic flexibility to characterize the multi-step transmission effects inherent in networked dynamical systems. 

Spectral estimation for network time series involves distinct methodological challenges. The high dimensionality inherent in large-scale networks requires regularization to ensure both computational feasibility and statistical reliability. However, standard tools such as the graphical Lasso \citep{hastie2009}—which are designed for real-valued covariance {\em selection} \citep{Dempster1972,lauritzen1996}—cannot be directly applied to the complex-valued spectral density and precision matrices fundamental to frequency-domain analysis. While recent developments, such as the complex graphical Lasso \citep{deb2024}, address this by constructing isomorphisms between complex-valued Hermitian matrices and real-valued representations, the problem of incorporating {\em known topological constraints} into these {\em frequency-domain objects} remains relatively unexplored. This work addresses these complexities by developing a framework where the network structure directly informs the spectral density and provides a robust pathway for high-dimensional inference.

By contrast to data-driven network discovery, in this work we assume a known network structure and {\em propose a spectral estimation framework for TS-NSD that directly incorporates the network topology}. To allow for potential departures from assumed modeling frameworks, we develop both parametric and nonparametric estimators while accounting for the network structure. Under parametric settings, unlike GSP-based methods, our approach leverages the structural properties of the broad class of network time series models including those discussed in \citet{Zhu2017}, \citet{Knight2020}, and \citet{Yin2023} to define a frequency-domain paradigm that naturally embeds the network. By treating these models as specialized, parsimonious representations of a VAR process, we establish network-specific spectral density and (partial) coherence measures. This parametric view is complemented by a nonparametric estimator that provides a robust alternative to parametric approaches, by directly constraining deviations from the network structure. Our estimators allow the inclusion of multi-hop neighborhood information, thus accommodating both direct and indirect temporal-spatial dependencies. This integration of network constraints into spectral analysis yields a flexible and interpretable methodology suited for complex, high-dimensional network time series data, which will be shown to be superior to classical spectral estimation via unrestricted VAR modeling. 


The paper is organized as follows. Section~\ref{sec:litrev} reviews the literature on TS-NSD establishing their unified representation as constrained VAR processes to link physical network topology with the stochastic dependencies used in spectral-domain analysis. Section~\ref{sec:Methods} defines the network spectral density and introduces the parametric and nonparametric estimation frameworks, whilst Section~\ref{NAR_spec_PN} proposes spectral estimators obtained via network constrained optimization. Section~\ref{sec:sim_results} evaluates the empirical performance of the proposed spectral estimation methods through extensive simulations. In Section~\ref{sec:application} we illustrate the utility of our framework for the motivating global bank network connectedness data, revealing underlying structures of systemic risk and providing new insights into how inter-bank dependence relationships vary across global financial networks. 


\section{Brief review of network time series modeling framework}\label{sec:litrev}

Formally, a \emph{network time series} \( X := (\mathbf{X}_t, G) \) is a stochastic process consisting of a multivariate time series $\{ \mathbf{X}_t \}_t\in \mathbb{R}^d$ recorded over the vertices (nodes) of an underlying network \( G = (K, E) \). Here, \( G \) is an undirected static graph with \( K \) denoting the set of $d$ nodes, and \( E \subseteq K \times K \) its set of edges. The distance between any pair of nodes \( (i,j) \) is denoted by \( \delta(i,j) \), and is given by the length of the shortest path between them. Following standard graph-theoretic conventions, an adjacency matrix \( \mathbf{A}\in \{0,1\}^{d \times d}  \) encodes the direct interactions among pairs of nodes in the graph \( G \), with a 1 marking each edge. 

Next we outline the primary characteristics of recently introduced models for such time series with network structured dependence (TS-NSD), where component dependencies are explicitly governed by an underlying graph. Specifically, we briefly review the GNAR framework of \cite{Knight2020} and the Network Autoregressive (NAR) models of \cite{Yin2023}, both of which offer alternative structures for modeling nodal interactions. Under such parametric modeling strategies, TS-NSD models can be represented as constrained VAR processes \citep{Lutkepohl2005}—a formulation sometimes referred to in the literature as stable network VAR models \citep{Li2025}. This connection establishes a fundamental link between network time series and graphical models which we subsequently pursue.

\subsection{The GNAR model of \cite{Knight2020}}\label{sec:gnarintro}

The $\text{GNAR}(p, [s_1, \dots, s_p])$ model defines the evolution of a nodal series $\{X_{i,t}\}$ as a function of its own (node $i$) history and the past values of its neighbors up to distance $s_k$ for each temporal lag $k=1, \ldots, p$, namely
\begin{equation}
    X_{i,t} = \sum_{k=1}^{p} \left(\alpha_{i,k} X_{i, t-k} + \sum_{r=1}^{s_k} \beta_{kr} \sum_{j \in \mathcal{N}_r(i)} w_{ij} X_{j,t-k} \right) + u_{i,t},
    \label{eq:GNAR}
\end{equation}
where $\alpha_{i,k}$ is the autoregressive parameter for node $i$ at lag $k$, $\beta_{kr}$ denotes the effect of the $r$-stage neighborhood $\mathcal{N}_r(i)$, and $\{u_{i,t}\}_t$ are independent and identically distributed white noise terms. A significant advantage of GNAR models is their parsimonious nature. In traditional VAR models, the number of parameters increases quadratically with the dimensionality $d$ and linearly with the order $p$. In contrast, the parameter space of a GNAR model is primarily determined by the depth of the $r$-stage neighborhood regression and the maximum lag. The specification in Equation~\eqref{eq:GNAR} allows for node-specific autoregressive coefficients, requiring the estimation of $(dp + \sum_{k=1}^p s_k)$ parameters. However, the framework is frequently adapted to a ``global-$\alpha$'' setting where $\alpha_{i,k} = \alpha_k$ for all nodes $i \in K$. In this case, the model contains only $p$ autoregressive coefficients and $\sum_{k=1}^{p} s_k$ neighborhood regression coefficients. This contrasts sharply with a standard $\text{VAR}(p)$ model, which requires $pd^2$ parameters. Provided that $p + \sum_{k=1}^p s_k < pd^2$, the GNAR framework ensures stable estimation and physical interpretability in high-dimensional TS-NSD settings where $d$ may be large relative to the observation period $T$. In the sequel, we focus on the global GNAR specification, though the local representation remains a valid extension.

In matrix form, the process may be represented as a constrained $\text{VAR}(p)$, $\mathbf{X}_t = \sum_{k=1}^{p} \mathbf{\Phi}_{k} \mathbf{X}_{t-k} + \mathbf{u}_{t}$, where the coefficient matrices $\mathbf{\Phi}_{k}$ are restricted by the network topology
\begin{equation}
    \mathbf{\Phi}_{k} = \mathbf{D}_{\alpha,k} + \sum_{r=1}^{s_k} \beta_{kr} (\mathbf{W} \odot \mathbf{A}_r).
    \label{eq:VAR_coeff}
\end{equation}
In the above, $\mathbf{D}_{\alpha,k} = \text{diag}(\alpha_{1,k}, \dots, \alpha_{d,k})$ denotes the diagonal matrix of node-specific autoregressive effects, $\mathbf{W} \in \mathbb{R}^{d \times d}$ provides the edge weights, while the sequence $\{ \mathbf{A}_r \}_{r \geq 1}$ represents the $r$-stage adjacency matrices; $\odot$ denotes the Hadamard product. Specifically, each entry in $\mathbf{A}_r$ is set to 1 if the corresponding node pair are $r$-stage neighbors, with $\mathbf{A}_1$ corresponding to the baseline adjacency matrix $\mathbf{A}$. This construction partitions the network influence according to shortest-path distance within each temporal lag, hence establishing a physically interpretable foundation for a spectral analysis.

\subsection{The NAR model of \cite{Yin2023}}\label{sec:nar_intro}

An alternative TS-NSD framework is the $\text{NAR}(q_1, q_2)$ model, which prioritizes nodal heterogeneity over the multi-stage spatial partitioning utilized in GNAR. In this model, $q_1$ and $q_2$ denote the number of self-lags and network-lags, respectively. The framework allows both the autoregressive and network effect parameters to vary across all nodes $i \in K$, resulting in the following nodal evolution
\begin{equation*}
    X_{i,t} = \sum_{j=1}^{q_1} \alpha_{i}^{(j)} X_{i, t-j} + \sum_{k=1}^{q_2} \beta_{i}^{(k)} \sum_{m=1}^{d} w_{im} X_{m,t-k} + u_{i,t},
    \label{eq:NAR_nodal_Yin}
\end{equation*}
where $\alpha_{i}^{(j)}$ and $\beta_{i}^{(k)}$ are node-specific self-lag and network-lag coefficients. 

Defining $q=\max\{q_1, q_2\}$, the process can be expressed in a compact vector form as $\mathbf{X}_t = \sum_{l=1}^{q} \mathbf{\Phi}_{l} \mathbf{X}_{t-l} + \mathbf{u}_{t}$. To maintain consistency with our notation so far, the transition matrices $\mathbf{\Phi}_l$ are structured as
\begin{equation}
    \mathbf{\Phi}_{l} = \mathbf{D}_{\alpha,l} + \mathbf{D}_{\beta,l} \mathbf{W},
    \label{eq:NAR_coeff_Yin}
\end{equation}
where $\mathbf{D}_{\alpha,l} = \text{diag}(\alpha_1^{(l)}, \dots, \alpha_d^{(l)})$ and $\mathbf{D}_{\beta,l} = \text{diag}(\beta_1^{(l)}, \dots, \beta_d^{(l)})$ are diagonal matrices containing the nodal parameters at lag $l$. 

A significant contribution of this framework is the derivation of a more flexible stability condition. Stationarity is established provided the spectral radius $\rho(\mathbf{\mathcal{A}}) < 1$, where $\mathbf{\mathcal{A}}$ denotes the companion form matrix of the VAR representation \citep{Lutkepohl2005}. This represents a broader and less restrictive condition than the absolute coefficient sums commonly used in earlier network literature. Furthermore, the model is designed to be robust under a broad class of innovations $\mathbf{u}_t$, explicitly accommodating dependent error structures such as spatial correlation or heteroskedasticity.

\subsection{Graphical models for TS-NSD}
\label{sec:graph_theory}

Many time-domain models for TS-NSD tend to focus on immediate network neighbors \citep{Zhu2017, Yin2023}. Under the graphical modeling framework of \cite{Dahlhaus2000}, if two nodes are not connected through an edge, their corresponding entry in the inverse spectral density matrix must be zero across all frequencies. However, in many real-world networks, nodes that lack a direct edge may still share information via multi-step paths \citep{Xu2016}. A notable exception, GNAR models \citep{Knight2020, nason2025} extend this concept by introducing interactions through $r$-stage neighborhood regressions for $r=1, \ldots, r^*$, where $r^* = \max\{s_1, \dots, s_p\}$ is the maximum active $r$-stage depth. This structure captures the gradual weakening of dependence as nodes become more distantly separated; the influence of an $r$-stage regression term diminishes as $r$ increases and typically becomes negligible when $r > r^*$. Consequently, if two nodes are sufficiently separated in the graph, their cross-dependence across all lags is expected to be minimal. The formal synthesis in \cite{nason2025} links this $r$-stage regression directly to the inverse spectral matrix. Two nodal time series are uncorrelated at all lags, conditioned on the rest of the network, if and only if they share no common active $r^*$-stage neighbors, a trait that we will subsequently use.

\section{Spectral inference for stationary TS-NSD}\label{sec:Methods}

We introduce the spectrum of TS-NSD in Section~\ref{sec:GNAR_Spec} and propose both parametric and nonparametric approaches for its estimation in Section~\ref{GNAR_Spec_Est}. Section~\ref{Par_GNAR_spec} discusses the consistency of the parametric spectral estimator, and an initial nonparametric estimation approach is presented in Section~\ref{sec:Np_AN}.

\subsection{TS-NSD spectral density matrix}
\label{sec:GNAR_Spec}

Should parametric modeling be appropriate, Section~\ref{sec:gnarintro} illustrated that network-structured time series models may be framed as constrained VAR processes, with specific structural constraints unique to each model class. Under this umbrella, the parametric spectral density \cite[see also][p.~685]{Priestley1981} is naturally defined by embedding the network structure directly into a $d$-dimensional VAR($p$) coefficient representation. 

\begin{dfn}[\textbf{TS-NSD spectra}]
\label{def:GNAR_Spec}
Let $\{\mathbf{X}_t\}$ be a stationary multivariate time series process defined on a network $G = (K,E)$. The parametric spectral density matrix is defined as the $d \times d$ complex-valued matrix 
\begin{equation}
   \mathbf{f}(\omega,G) = \mathbf{U}_{p}^{-1}(\omega,G) \mathbf{V}_{d} \left(\overline{\mathbf{U}}_{p}^{-1}(\omega,G)\right)^\top,
    \label{GNAR_Spec}
\end{equation}
where $\mathbf{V}_d$ denotes the covariance matrix of the innovation process $\{\mathbf{u}_{t}\}$ and $\overline{\cdot}$ indicates the complex conjugate; for example, the matrix $\mathbf{U}_{p}(\omega,G)$ in the specific case of the global$-\alpha$ GNAR model is given by
\begin{equation}
   \mathbf{U}_{p}(\omega,G) = \mathds{I}_d - \sum_{k=1}^{p} \left[\alpha_k \mathds{I}_d + \sum_{r=1}^{s_k} \beta_{kr} (\mathbf{W} \odot \mathbf{A}_r) \right] e^{-i 2\pi k \omega},
   \label{Up}
\end{equation}
for any frequency $\omega \in (0,0.5]$. The use of the second argument, $\mathbf{f}(\cdotp,G)$, clarifies that the network $G$ is embedded in the defined spectral quantities.
\end{dfn}

\begin{remark}
Although the above focuses on the GNAR framework as an {\em option} for parametric estimation, similarly following the VAR representations in \eqref{eq:VAR_coeff} and \eqref{eq:NAR_coeff_Yin}, these concepts can be readily adapted to other network-based models such as \cite{Zhu2017,Yin2023}.
\end{remark}

Definition \ref{def:GNAR_Spec} extends the classical VAR spectral representation by incorporating network dependencies through the edge-weight matrix $\mathbf{W}$ and the $r$-stage adjacency matrices $\mathbf{A}_r$, thereby effectively capturing the structural influence of the graph on the spectral properties of the process. Moreover, based on the parametric spectrum in Definition \ref{def:GNAR_Spec}, the population versions of squared coherence and squared partial coherence can be defined as in the following equations, \eqref{eq:coh} and \eqref{eq:pcoh}, with the latter relying on the inverse of the parametric spectrum, $\mathbf{S}(\omega,G) := \mathbf{f}(\omega,G)^{-1}$,

\begin{equation}\label{eq:coh}
[\boldsymbol{\rho}(\omega,G)]_{ij}^2 = 
\frac{ \left|  [\mathbf{f}(\omega,G)]_{ij} \right|^2 }
{  [\mathbf{f}(\omega,G)]_{ii} \cdot [\mathbf{f}(\omega,G)]_{jj}  }, 
\quad i,j=1,\ldots,d,
\end{equation}

\begin{equation}\label{eq:pcoh}
[\boldsymbol{\gamma}(\omega,G)]_{ij}^2 = 
\frac{ \left|  [\mathbf{S}(\omega,G)]_{ij} \right|^2 }
{  [\mathbf{S}(\omega,G)]_{ii} \cdot [\mathbf{S}(\omega,G)]_{jj}  }, 
\quad i,j=1,\ldots,d.
\end{equation}

\subsection{TS-NSD spectral density estimation}\label{GNAR_Spec_Est}

Various estimators for the TS-NSD spectral density matrix may be proposed in order to estimate $\mathbf{f}(\cdotp,G)$.  Definition~\ref{def:GNAR_Spec} suggests a natural starting point to be the parametric estimator that utilises the network time series model parameters and the covariance structure of the innovation process.

\subsubsection{Parametric estimation of TS-NSD spectrum}\label{Par_GNAR_spec}

The parametric estimation of the TS-NSD spectral density matrix $\mathbf{f}(\cdotp,G)$ is derived from its definition in~\eqref{GNAR_Spec}, coupled with estimates of the model parameters in e.g.,~\eqref{Up}, and involves the following steps:

\begin{enumerate}[(a)]
   \item \textbf{Model fitting:} The TS-NSD model coefficients can be estimated using least squares or penalized regression techniques, while incorporating the structure of the underlying network. For instance, the global-$\alpha$ GNAR model coefficients, e.g. $\{\alpha_k\}_{k}$ and $\{\beta_{kr}\}_{k,r}$ can be estimated by expressing the process as a linear model fitted by ordinary (OLS) or generalized least squares (GLS), while the NAR model of \cite{Yin2023} employs GLS or its empirical counterpart (EGLS) to estimate the node-specific parameters, or a ridge-regularized GLS estimator when the number of nodes exceeds the number of time points. Furthermore, methods for order selection to ensure proper model specification are explored in \cite{Knight2020}.

   \item \textbf{Residual covariance estimation:} The covariance matrix of the innovation process, $\mathbf{V}_d$, is estimated from the residuals of the fitted TS-NSD model. In the case of GNAR, Appendix B of \cite{Knight2020} presents a consistent estimator for the innovation covariance matrix adapted from \cite{Lutkepohl2005}. \cite{Mantziou2023} advocate a similar strategy for GNAR-edge fitting. Errors that exhibit heavier tails than Gaussian and a variety of error covariance structures are adapted by \cite{Yin2023}.
    
    \item \textbf{TS-NSD spectral matrix estimation:} The estimated parameters are substituted into the spectral representation in~\eqref{GNAR_Spec}, leading to the parametric estimator:
    \begin{equation}\label{eq:gnarspecest}
        \widehat{\mathbf{f}}(\omega,G) = \widehat{\mathbf{U}}_p^{-1}(\omega,G) \widehat{\mathbf{V}}_d \left( \widehat{\overline{\mathbf{U}}}_p^{-1}(\omega,G) \right)^\top,
    \end{equation}
    where $\widehat{\mathbf{U}}_p(\omega,G)$ is model-dependent and uses the estimated TS-NSD parameters, e.g. see for example~\eqref{Up}.
\end{enumerate}

\noindent{\bf Theoretical properties for the parametric TS-NSD spectrum estimator.}
The desirable behaviour of the parametric estimator is naturally derived from the consistency of the TS-NSD model parameter estimates under standard regularity conditions. 

\begin{proposition}\label{prop:GNAR_Spec_Est}
Let $\{\mathbf{X}_t\}$ be a stationary multivariate time series process defined on a static network $G = (K,E)$. When the true TS-NSD process is correctly specified as in Definition~\ref{def:GNAR_Spec} with independent white noise innovations with finite fourth moment and covariance matrix ${\mathbf{V}}_d$,  the parametric spectral density matrix in equation~\eqref{eq:gnarspecest} employing the associated (GLS) estimators is consistent for the truth,  
\begin{equation}
     \widehat{\mathbf{f}}(\omega,G) \xrightarrow{p} \mathbf{f}(\omega,G), \text{ for all frequencies } \omega \in (0,0.5].
\end{equation}
\end{proposition}
\begin{proof}
The result follows from a straightforward coupling of the estimator properties tailored for~\eqref{Up} and detailed in Propositions 1 and 2 in Appendix~B of \cite{Knight2020}, and the continuous mapping theorem.
\end{proof}

While parametric estimation is computationally efficient and interpretable, its accuracy relies on the correct specification of the model. To mitigate the risk of model misspecification, a nonparametric estimation approach can also be considered, as introduced next. 

\subsubsection{Nonparametric estimation of the TS-NSD spectrum}\label{sec:Np_AN}

Since a TS-NSD process is inherently a multivariate time series, traditional nonparametric spectral density estimation techniques, such as periodogram-based estimators, can be obtained. We first recall the usual construction of a nonparametric spectral estimator, that does {\em not} incorporate network-specific structure. The procedure follows standard multivariate time series methods and involves the primary steps below:

\begin{enumerate}[(i)]
    \item \textbf{Fourier raw periodogram:}\label{rawfourier} Consider a stationary TS-NDS process $\{\mathbf{X}_t\}_{t=1}^T$ over the network $G = (K,E)$ with $d$ nodes. The discrete Fourier transform of $\{\mathbf{X}_t\}$ is 
    \begin{equation*}
        \mathbf{J}(\omega_l) = \frac{1}{\sqrt{T}} \sum_{t=1}^{T} \mathbf{X}_t e^{-i 2 \pi t \omega_l},
    \end{equation*}
    where $\omega_l = l/T$ for $l =0, \, 1, \ldots n_T=([T/2]-1)$ are the evaluation Fourier frequencies. The $d \times d$ periodogram matrix, which serves as a raw estimate of the spectral density, is defined as 
    \begin{equation}
        \mathbf{I}_T(\omega_l) = \mathbf{J}(\omega_l) \overline{\mathbf{J}}(\omega_l)^{\top},
        \label{eq:period}
    \end{equation}
     with its asymptotic behaviour studied in depth across the classical Fourier literature, see e.g.,  Section 11.7 in \cite{Brockwell1992}. Under the assumption that the process $\{\mathbf{X}_t\}$ is Gaussian with an absolutely summable autocovariance function (to which we henceforth refer as Assumption A), asymptotically $\mathbf{J}(\omega_l) \dot\sim N_{\mathbb{C}} \left(\mathbf{0}, \mathbf{f}(\omega_l) \right)$ holds, approximately independent at different Fourier frequencies. 
     
     \begin{remark} Here, we use the notation $\mathbf{f}(\cdotp)$ for the spectral matrix of $\{\mathbf{X}_t\}$ rather than $\mathbf{f}(\cdotp, G)$, in order to emphasize that the network structure underpinning the process is not captured.
     \end{remark}

    \item \textbf{Smoothed periodogram:} \label{sec:smoothper} To obtain a consistent estimator of the network-agnostic spectral matrix $\mathbf{f}$, the periodogram needs to be smoothed. Given the bandwidth parameter $m_T=o(\sqrt T)$ and weight function $\{W_T(\cdot)\}$, a smoothed spectral estimator $\widetilde{\mathbf{I}}_T(\omega_l)$ is defined as
    \begin{equation}
\widetilde{\mathbf{I}}_T(\omega_l) = \frac{1}{2 m_T + 1} \sum_{|k| \leq m_T} W_T(k) \mathbf{I}_T\left(\omega_{(l+k) {\text{mod}(T)}}\right), 
        \label{eq:smooth_per}
    \end{equation}
    where the subscript notation in $\omega_{(l+k)mod(T)}$ is to be understood as evaluating the indices modulo $T$.
    For guidance on the bandwidth and weight choices, see \cite{Priestley1981}. The consistency of the smoothed periodogram is discussed in e.g., \cite{Brockwell1992,Shumway2017}. 

    \begin{remark} Under Assumption A, since the Fourier coefficients are asymptotically $\mathbf{J}(\omega_l) \dot\sim N_{\mathbb{C}} \left(\mathbf{0}, \mathbf{f}(\omega_l) \right)$  at every frequency $\omega_l$ and approximately independent across frequencies,   
    an approximation to the (pseudo) log-likelihood can be constructed and used in order to estimate the spectral density and precision under e.g., further sparsity constraints \citep{deb2024}. The log-likelihood function for the precision matrix $\mathbf{f}^{-1}(\cdot)$ can be written at each frequency $\omega_l$ (up to a constant) as
    \begin{equation}\label{eq:loglik}
     \mathcal{L}(\mathbf{f}^{-1}(\omega_l))= \log \det \mathbf{f}^{-1}(\omega_l) - \operatorname{tr} \left( \mathbf{f}^{-1}(\omega_l) \widetilde{\mathbf{I}}_T(\omega_l)\right),
    \end{equation}    
where $\operatorname{tr}$ denotes the trace operator.
\end{remark}
    
\end{enumerate}

\section{Spectral estimation via network constrained optimization}\label{NAR_spec_PN}
While the smoothed periodogram in Section~\ref{sec:Np_AN}~\ref{sec:smoothper}, as well as traditional parametric VAR-type spectral estimators \citep{Priestley1981}, provide flexible data-driven approaches, they remain fundamentally agnostic to the network topology underpinning TS-NSD processes. We address this limitation by introducing spectral estimators constrained by the underlying adjacency matrix structure that effectively integrate the \textit{known} network connectivity into the estimation procedure in Section~\ref{GNAR_spec_PN}, where estimator consistency is also derived. Section~\ref{sec:GNAR_induced} crucially proposes a framework that allows for the consistent estimation of {\em multi-hop} induced interdependencies in the spectral domain.

\subsection{Constrained spectral estimation via the adjacency matrix}\label{GNAR_spec_PN}
Classical likelihood-based estimation methods under network-driven structural constraints \citep{Songsiri2009} are generally designed for real-valued matrices, whereas the smoothed periodogram is complex-valued. To bridge this gap, akin to \cite{Fiecas2019, deb2024} albeit in our setting we are not driven by network learning but rather by network-structure enforcing, we propose to reformulate the complex-valued constrained optimization spectral estimation problem into an equivalent real-valued setting. We illustrate this framework using the smoothed periodogram $\widetilde{\mathbf{I}}_T(\cdotp)$, however the procedure is general and can be applied to any network-agnostic spectral estimator, including parametric VAR-type estimators.

To achieve this, at each Fourier frequency \( \omega_l \) we decompose the complex-valued spectral estimator \( \widetilde{\mathbf{I}}_T(\omega_l) \in \mathds{C}^{d \times d} \) into its real and imaginary components, denoted by \( \widetilde{\mathbf{C}}(\omega_l) \in \mathds{R}^{d \times d} \) and \( \widetilde{\mathbf{Q}}(\omega_l) \in \mathds{R}^{d \times d} \) respectively, and construct the real-valued augmented matrix
\begin{equation}
\widetilde{\mathbf{\Sigma}}(\omega_l) = {\frac{1}{2}}
\begin{bmatrix}
\widetilde{\mathbf C}(\omega_l) & -\widetilde{\mathbf Q}(\omega_l) \\
\widetilde{\mathbf Q}(\omega_l) & \widetilde{\mathbf C}(\omega_l)
\end{bmatrix}
\label{eq:tildeSigma}
\end{equation}
that serves as a consistent estimator of the population analogue \( \mathbf{\Sigma}(\omega_l) \), which can be viewed as the covariance matrix of the transformed real-valued discrete Fourier coefficients (see Appendix~\ref{app:1} for its derivation).

To introduce the network-based structural constraints, we use the explicit graph edge--spectral matrix connection (see Section~\ref{sec:graph_theory}) and enforce zeros in the inverse covariance matrix \( {\mathbf \Theta}(\omega_l) := {\mathbf \Sigma}^{-1} (\omega_l)\) by means of the adjacency matrix \( \mathbf{A}_1 \) of the physical network \( G = (K, E) \). Since \( {\mathbf \Sigma}(\omega_l) \in \mathds{R}^{2d \times 2d} \), we define an augmented adjacency matrix
\begin{equation}
   \tilde{\mathbf{A}}_1 =
\begin{bmatrix}
\mathbf{A}_1 & \mathbf{A}_1 \\
\mathbf{A}_1 & \mathbf{A}_1
\end{bmatrix},
\label{eq:A1}
\end{equation}
which ensures the structural constraints are adhered to consistently across both real and imaginary blocks. The estimation can then be reformulated as a constrained optimization problem by maximizing the constrained Gaussian log-likelihood (see also~\eqref{eq:loglik}) 
\begin{equation}
\ell(\mathbf{\Theta}) = \log \det \mathbf{\Theta} - \operatorname{tr}(\widetilde{\mathbf{\Sigma}} \mathbf{\Theta}) - \sum_{\{(i,j)\,|\,(\tilde{\mathbf{A}}_1)_{i,j}=0 \}} \lambda_{ij} \theta_{ij},
\label{eq:constrained_likelihood}
\end{equation}
where \( \lambda_{ij} \) are Lagrange multipliers that penalize non-zero entries corresponding to absent edges in the network and for brevity we dropped the frequency dependence. The resulting first-order condition, $\mathbf{\Theta}^{-1} - \widetilde{\mathbf \Sigma} - \mathbf \Lambda = 0$, is solved using iterative regression updates \citep{hastie2009} to obtain the refined precision estimator \( \widehat{\widetilde{\mathbf \Theta}} \) and its inverse \( \widehat{\widetilde{\mathbf \Sigma}} \). Finally, the complex-valued {\em network-constrained} spectral estimator is reconstructed as
\begin{equation}
\widehat{\tilde{\mathbf{f}}}(\omega_l,G):= \widehat{\widetilde{\mathbf{I}}}(\omega_l) = \widehat{\widetilde{\mathbf C}}(\omega_l) - i \widehat{\widetilde{\mathbf Q}}(\omega_l).
\label{eq:GNAR_NP_PN}
\end{equation}

\begin{proposition}\label{prop:Nonpar_Spec_Est}
Let $\{\mathbf{X}_t\}$ be a stationary multivariate time series process defined on a static network $G = (K,E)$ which satisfies Assumption A--Gaussian process with absolutely summable autocovariance function. Then the nonparametric constrained spectral density matrix estimator in equation~\eqref{eq:GNAR_NP_PN}, that maximizes~\eqref{eq:constrained_likelihood}, is consistent for the truth, 
\begin{equation}
     \widehat{\tilde{\mathbf{f}}}(\omega_l,G) \xrightarrow{p} \mathbf{f}(\omega_l,G), \, \forall \omega_l \in (0,0.5].
\end{equation}
\end{proposition}
\begin{proof}
Theoretical properties of this  estimator, including its connection to maximum likelihood theory for covariance selection \citep{Dempster1972, lauritzen1996}, are detailed in Appendix~\ref{app:B}.
\end{proof}

\subsection{Constrained spectral estimation embedding higher-order graph interactions} \label{sec:GNAR_induced}

The spectral estimation introduced in Section~\ref{GNAR_spec_PN} relies on solving the maximization problem~\eqref{eq:constrained_likelihood} subject to constraints   derived from the physical adjacency matrix ${\mathbf{A}}_1$, through its extension~\eqref{eq:A1} and crucially, this assumes that conditional dependence exists only between nodes sharing a {\em direct edge}. However, in many complex systems, the data incorporate higher-order interactions where nodes share information via {\em intermediate paths} \citep{Giraldo2025}. In such scenarios, a penalty based strictly on $\mathbf{A}_1$ is misspecified, as it forces the estimator to suppress valid statistical dependencies between non-adjacent nodes.

To address this, we propose a {\em higher-order} adjacency matrix, $\mathbf{A}_{\text{HO}}$, to inform the optimization. While various network time series parametric models can be used as base for constructing this matrix, the GNAR framework provides a rigorous foundation for capturing such dependencies. Specifically, a GNAR process of maximum depth $r^*$ induces a correlation structure where two nodal time series $\{X_{i,t}\}$ and $\{X_{j,t}\}$ become conditionally uncorrelated at all lags if and only if their distance in the underlying network $G$ satisfies $\delta_{G}(i,j) > 2r^*$. This implies that the induced graph---representing non-zero entries in the precision matrix---is naturally denser than the physical network $G$.

Following this logic, we define the higher-order penalty matrix as
\begin{equation}
    \mathbf{A}_{\text{HO}} = \sum_{r=1}^{\min\{2r^*, r_{\text{max}}\}} \mathbf{A}_r,
\end{equation}
where $r_{\text{max}}$ denotes the network diameter, which simply put is the length of the longest shortest path in $G$. This formulation ensures the penalty only suppresses correlations between nodes that are truly independent under an $r^*$-stage dependency structure. For instance, if $r^*=1$, $\mathbf{A}_{\text{HO}} = \mathbf{A}_1 + \mathbf{A}_2$, capturing both direct and neighbor-of-neighbor interactions. Its mapping from complex- to real-valued $\tilde{\mathbf{A}}_{\text{HO}}$ can be carried out analogously to $\tilde{\mathbf{A}}_1$ in equation~\eqref{eq:A1}, and the penalized spectral estimator will be the solution to a problem akin to~\eqref{eq:constrained_likelihood} with $\tilde{\mathbf{A}}_{\text{HO}}$ replacing $\tilde{\mathbf{A}}_{1}$. Other TS-NSD models, such as \cite{Yin2023}, could similarly define $\mathbf{A}_{\text{HO}}$ by utilizing their specific weighting schemes to identify relevant higher-order paths.

This framework is primarily designed to inform spectral estimators with network connectivity knowledge, and both parametric and nonparametric estimators agnostic of the network could benefit from penalization via $\mathbf{A}_{\text{HO}}$ to capture potential higher-order interactions. However, capturing these interactions introduces a trade-off regarding model parsimony: as the neighborhood depth $r^*$ increases, $\mathbf{A}_{\text{HO}}$ becomes less sparse; in the limit where $2r^* \geq r_{\text{max}}$, the matrix becomes fully connected and the penalization loses its purpose, as it no longer enforces structural constraints. We explore the balance between capturing higher-order dependencies and maintaining meaningful penalization in Section~\ref{sec:sim_results}.

\section{Simulation study}\label{sec:sim_results}

In this section, we introduce a series of simulation scenarios designed to evaluate the estimation methods presented in Sections~\ref{sec:Methods}~and ~\ref{NAR_spec_PN}. The primary objective is to assess the performance of these estimators by examining the impact of sample size \( T \), different network structures \( G \), and potential model misspecification. We focus on networks with \( d=5 \) and \( 10 \) nodes, generated using the \pkg{GNAR} \Rlogo package (see Figure~\ref{fig:Net_1}), to evaluate seven spectral estimation approaches. As noted in Section~\ref{sec:gnarintro}, various TS-NSD models could be utilized as parametric estimators; however, we employ the GNAR framework here due to its computational efficiency and parsimonious parameterization. For instance, \cite{Yin2023} showed that a global-$\alpha$ GNAR$(1,[5])$ in a network of 100 nodes requires the estimation of only 6 model parameters, whereas a NAR$(1,1)$ model allowing for nodal heterogeneity requires at least 200 parameters, yet their reported mean square error results remain highly competitive for the GNAR model. 

Consequently, we evaluate seven estimation methods (EM). The parametric approaches include: the baseline parametric spectrum (\textbf{EM1}: \textbf{Par}); a VAR spectrum agnostic to the network (\textbf{EM2}: \textbf{PVAR}); and VAR spectrum penalized by the higher-order interaction adjacency matrix $\mathbf{A}_{\text{HO}}$ described in Section~\ref{sec:GNAR_induced} (\textbf{EM3}: \textbf{PVAR\_{\text{HO}}}) and the physical adjacency matrix $\mathbf{A}_1$ described in Section~\ref{GNAR_spec_PN} (\textbf{EM4}: \textbf{PVAR\_A1}). Similarly, we assess nonparametric estimators penalized by $\mathbf{A}_{\text{HO}}$ (\textbf{EM5}: \textbf{NP\_{\text{HO}}}) and $\mathbf{A}_1$ (\textbf{EM6}: \textbf{NP\_A1}), alongside a network-agnostic nonparametric version (\textbf{EM7}: \textbf{NP}).

Table~\ref{tab:GNAR_params} presents the true GNAR parameters and model orders for the simulation scenarios considered in Sections~\ref{sim:no_miss} and \ref{sim:miss}, designed to represent a variety of process behaviors, explored under Sections~\ref{sim:no_miss}-- \ref{sim:miss}.

\begin{figure}
    \centering
    \begin{minipage}{0.48\linewidth}
        \centering
        \includegraphics[width=\linewidth]{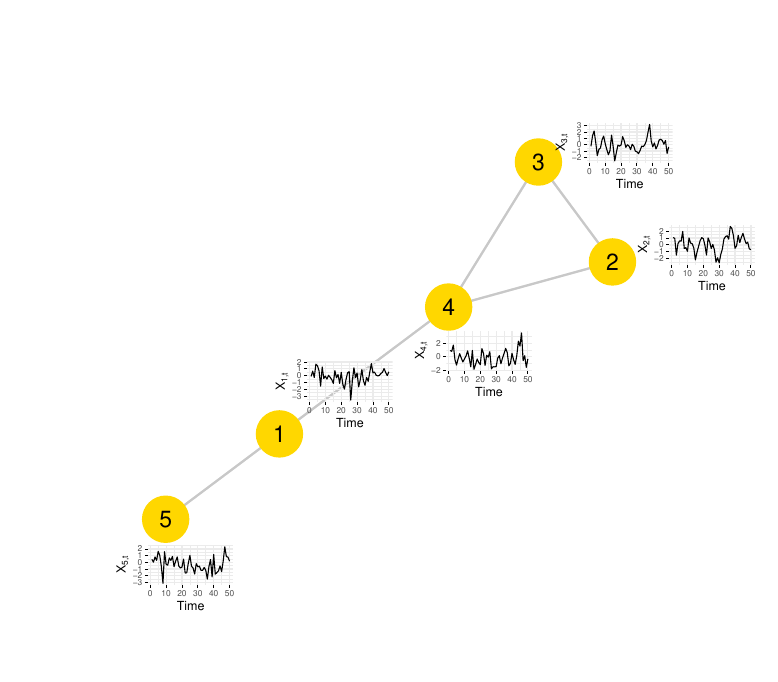}
    \end{minipage}
    \hfill
    \begin{minipage}{0.48\linewidth}
        \centering
        \includegraphics[width=\linewidth]{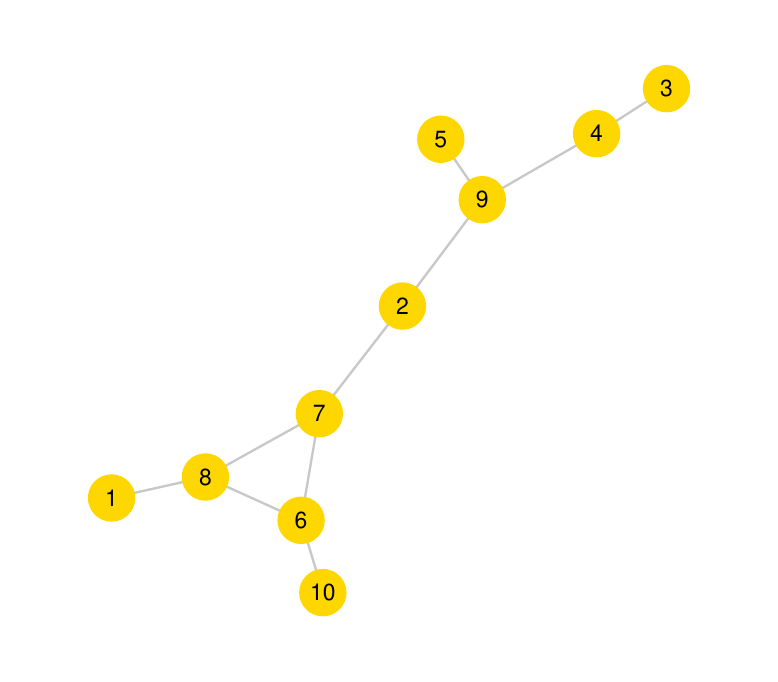}
    \end{minipage}
    \caption{Networks with 5 and 10  nodes network generated with the \pkg{GNAR} \Rlogo package.}
    \label{fig:Net}
    \label{fig:Net_1}
\end{figure}

\begin{table}[h]
    \centering
    \caption{Parameter settings for GNAR simulations}
    \label{tab:GNAR_params}
    \begin{tabular}{lcccc}
        \hline
        Model & $\alpha_1$ & $\alpha_2$ & $\alpha_3$ & $\beta$ \\
        \hline
        \textbf{M1:}GNAR(2, [1,1]) & 0.2 & 0.2 & -- & (0.2, 0.1) for 1-stage \\
         \textbf{M2:}GNAR(2, [1,2]) & 0.1 & 0.1 & -- & (0.075 for 1-stage, 0.05, 0.15 for 2-stage) \\
         \textbf{M3:}GNAR(2, [2,3]) & 0.2 & 0.1 & -- & (0.075, 0.05 for 2-stage, 0.05, 0.05, 0.1 for 3-stage) \\
        \textbf{M4:}GNAR(3, [1,2,3]) & 0.1 & 0.075 & 0.05 & (0.1 for 1-stage, 0.075 for 2-stage, 0.05 for 3-stage) \\
         \textbf{M5:}GNAR(3, [3,3,3]) & 0.15 & 0.1 & 0.05 & (0.05 for 1-stage, 0.05 for 2-stage, 0.05 for 3-stage) \\
        \hline
    \end{tabular}
\end{table}

\subsection{Spectral recovery under exact structural alignment}\label{sim:no_miss}

In this section, we evaluate the precision of spectral estimators in a controlled setting where the underlying network process aligns perfectly with the parametric assumptions. Specifically, the \textbf{EM1 (Par)} method assumes that the model orders—the temporal lag and the $r$-stage neighborhood depth—are known, as detailed in Table~\ref{tab:GNAR_params}. This setup allows us to isolate the performance of spectral-related quantities (coherence and partial coherence) by removing errors stemming from model selection or parameter misspecification.

To characterize the frequency-domain dependencies across the network nodes, we estimate the squared coherence and squared partial coherence (defined in equations~\eqref{eq:coh} and \eqref{eq:pcoh}, respectively). These quantities are central to our framework, providing a dynamic view of node-to-node interactions across frequencies. Asymptotically, these estimators exhibit desirable properties, with consistency following from Propositions~\ref{prop:GNAR_Spec_Est} and~\ref{prop:Nonpar_Spec_Est} by a direct application of Slutsky's theorem \citep{Slutsky1925}. 


The accuracy of the seven estimation methods is quantified using the Root Mean Square Error (RMSE) based on the Frobenius norm:
\begin{equation}
    \mathrm{RMSE} = \left( \frac{1}{R \cdot n_T} \sum_{rep=1}^R \sum_{\ell=1}^{n_T} \left\| \widehat{\mathbf{f}}^{(rep)}(\omega_\ell,G) - \mathbf{f}(\omega_\ell,G) \right\|^2 \right)^{1/2},
    \label{eq:RMSE}
\end{equation}
where $R = 500$ independent realizations are used for each model across sample sizes $T \in \{100, 200, 500, 1000\}$. Tables~\ref{Sim:Table_no_miss_coh} and~\ref{Sim:Table_no_miss_pcoh} present the RMSE results for coherence and partial coherence, respectively (spectral density RMSE is provided in Appendix~\ref{app:C}).

The results demonstrate that when the structural model assumptions are correct, the parametric estimator consistently outperforms nonparametric alternatives. Notably, in larger, sparser networks, penalization further refines the recovery of the spectral matrix. Methods penalized by the network's adjacency matrix typically outperform those penalized by the higher-order correlation structure. However, in the five-node network, where the diameter is smaller than the $r^* = 3$ neighborhood depth, penalization provides no significant benefit to the VAR or nonparametric methods. The superior performance of the parametric approach is particularly evident in low-sample-size regimes ($T=100$), where it leverages the network structure to provide stable estimates regardless of model complexity.

\setlength{\tabcolsep}{3pt} 
\renewcommand{\arraystretch}{0.9} 
\begin{longtable}{|c||*{7}{c|}|*{7}{c|}}
\caption{RMSE $\times 100$ for the coherence estimates obtained using seven estimation methods (EM1–EM7) across five models (M1–M5). Results are presented for both five-node and ten-node networks under increasing time series lengths $T = 100, 200, 500, 1000$.} \label{Sim:Table_no_miss_coh}\\ 
\hline
\multirow{2}{*}{\textbf{Model}} 
& \multicolumn{7}{c|}{\textbf{T = 100}} 
& \multicolumn{7}{c|}{\textbf{T = 200}} \\
\cline{2-15}
& \textbf{EM1} & \textbf{EM2} & \textbf{EM3} & \textbf{EM4} & \textbf{EM5} & \textbf{EM6} & \textbf{EM7} 
& \textbf{EM1} & \textbf{EM2} & \textbf{EM3} & \textbf{EM4} & \textbf{EM5} & \textbf{EM6} & \textbf{EM7} \\
\hline
\endfirsthead
\multicolumn{15}{c}{\textit{(Table \ref{Sim:Table_no_miss_coh} continued - T=100, 200)}} \\
\hline
\multirow{2}{*}{\textbf{Model}} & \multicolumn{7}{c|}{\textbf{T = 100}} & \multicolumn{7}{c|}{\textbf{T = 200}} \\
\cline{2-15}
& \textbf{EM1} & \textbf{EM2} & \textbf{EM3} & \textbf{EM4} & \textbf{EM5} & \textbf{EM6} & \textbf{EM7} & \textbf{EM1} & \textbf{EM2} & \textbf{EM3} & \textbf{EM4} & \textbf{EM5} & \textbf{EM6} & \textbf{EM7} \\
\hline
\endhead
\multicolumn{15}{r}{\textit{(continued on next page)}} \\
\endfoot
\endlastfoot

\multicolumn{15}{c}{\textit{Five Nodes Network}} \\
\hline
\textbf{M1} & 0.97 & 5.82 & 5.27 & 4.29 & 8.25 & 6.56 & 9.17 & 0.73 & 3.19 & 2.95 & 2.51 & 6.03 & 4.91 & 6.63 \\
\textbf{M2} & 1.11 & 5.74 & 5.74 & 4.27 & 8.70 & 6.19 & 8.70 & 0.79 & 3.21 & 3.21 & 2.80 & 6.39 & 4.69 & 6.39 \\
\textbf{M3} & 1.26 & 5.54 & 5.54 & 5.54 & 8.63 & 8.63 & 8.63 & 0.88 & 3.04 & 3.04 & 3.04 & 6.25 & 6.25 & 6.25 \\
\textbf{M4} & 1.48 & 8.22 & 8.22 & 8.22 & 8.73 & 8.73 & 8.73 & 1.05 & 4.43 & 4.43 & 4.43 & 6.40 & 6.40 & 6.40 \\
\textbf{M5} & 1.81 & 8.32 & 8.32 & 8.32 & 8.93 & 8.93 & 8.93 & 1.19 & 4.32 & 4.32 & 4.32 & 6.49 & 6.49 & 6.49 \\
\hline
\multicolumn{15}{c}{\textit{Ten Nodes Network}} \\
\hline
\textbf{M1} & 0.50 & 7.13 & 5.04 & 3.65 & 6.00 & 5.00 & 6.60 & 0.40 & 2.70 & 2.50 & 2.10 & 4.20 & 3.60 & 4.90 \\
\textbf{M2} & 0.60 & 6.90 & 5.50 & 3.60 & 6.80 & 5.10 & 6.80 & 0.45 & 2.75 & 2.75 & 2.30 & 4.40 & 3.55 & 4.60 \\
\textbf{M3} & 0.70 & 6.60 & 6.00 & 5.50 & 6.60 & 6.00 & 6.60 & 0.55 & 2.60 & 2.60 & 2.60 & 4.20 & 4.20 & 4.40 \\
\textbf{M4} & 0.85 & 8.40 & 8.00 & 7.50 & 6.90 & 6.50 & 6.90 & 0.70 & 3.80 & 3.80 & 3.80 & 4.50 & 4.20 & 4.40 \\
\textbf{M5} & 1.00 & 8.50 & 8.20 & 8.10 & 7.10 & 6.80 & 7.00 & 0.85 & 3.60 & 3.60 & 3.60 & 4.60 & 4.30 & 4.30 \\
\hline
\end{longtable}

\vspace{-1.5em} 
\addtocounter{table}{-1} 

\begin{longtable}{|c||*{7}{c|}|*{7}{c|}}
\caption{\textit{(continued)}} \\ 
\hline
\multirow{2}{*}{\textbf{Model}} 
& \multicolumn{7}{c|}{\textbf{T = 500}} 
& \multicolumn{7}{c|}{\textbf{T = 1000}} \\
\cline{2-15}
& \textbf{EM1} & \textbf{EM2} & \textbf{EM3} & \textbf{EM4} & \textbf{EM5} & \textbf{EM6} & \textbf{EM7} 
& \textbf{EM1} & \textbf{EM2} & \textbf{EM3} & \textbf{EM4} & \textbf{EM5} & \textbf{EM6} & \textbf{EM7} \\
\hline
\endfirsthead
\multicolumn{15}{c}{\textit{(Table \ref{Sim:Table_no_miss_coh} continued - T=500, 1000)}} \\
\hline
\multirow{2}{*}{\textbf{Model}} & \multicolumn{7}{c|}{\textbf{T = 500}} & \multicolumn{7}{c|}{\textbf{T = 1000}} \\
\cline{2-15}
& \textbf{EM1} & \textbf{EM2} & \textbf{EM3} & \textbf{EM4} & \textbf{EM5} & \textbf{EM6} & \textbf{EM7} & \textbf{EM1} & \textbf{EM2} & \textbf{EM3} & \textbf{EM4} & \textbf{EM5} & \textbf{EM6} & \textbf{EM7} \\
\hline
\endhead
\multicolumn{15}{r}{\textit{(continued on next page)}} \\
\endfoot
\bottomrule
\endlastfoot

\multicolumn{15}{c}{\textit{Five Nodes Network}} \\
\hline
\textbf{M1} & 0.47 & 1.59 & 1.51 & 1.35 & 3.98 & 3.28 & 4.36 & 0.34 & 1.02 & 0.98 & 0.92 & 2.91 & 2.46 & 3.15 \\
\textbf{M2} & 0.53 & 1.62 & 1.62 & 2.08 & 4.24 & 3.40 & 4.24 & 0.35 & 1.01 & 1.01 & 1.89 & 3.08 & 2.75 & 3.08 \\
\textbf{M3} & 0.55 & 1.52 & 1.52 & 1.52 & 4.16 & 4.16 & 4.16 & 0.40 & 0.95 & 0.95 & 0.95 & 3.35 & 3.25 & 3.60 \\
\textbf{M4} & 0.65 & 1.80 & 1.80 & 1.80 & 4.50 & 4.20 & 4.50 & 0.45 & 1.05 & 1.05 & 1.05 & 3.55 & 3.40 & 3.60 \\
\textbf{M5} & 0.70 & 1.85 & 1.85 & 2.00 & 4.80 & 4.50 & 4.70 & 0.50 & 1.10 & 1.10 & 1.10 & 3.75 & 3.55 & 3.80 \\
\hline
\multicolumn{15}{c}{\textit{Ten Nodes Network}} \\
\hline
\textbf{M1} & 0.38 & 1.58 & 1.30 & 1.20 & 3.45 & 3.00 & 3.30 & 0.30 & 1.05 & 0.95 & 0.85 & 2.85 & 2.50 & 3.00 \\
\textbf{M2} & 0.45 & 1.72 & 1.55 & 1.40 & 3.80 & 3.10 & 3.80 & 0.35 & 1.15 & 1.05 & 0.95 & 3.10 & 2.60 & 3.20 \\
\textbf{M3} & 0.50 & 1.80 & 1.70 & 1.60 & 3.90 & 3.50 & 3.90 & 0.40 & 1.20 & 1.10 & 1.05 & 3.25 & 2.80 & 3.40 \\
\textbf{M4} & 0.65 & 2.10 & 2.00 & 1.90 & 4.20 & 3.80 & 4.20 & 0.50 & 1.35 & 1.25 & 1.20 & 3.45 & 3.10 & 3.50 \\
\textbf{M5} & 0.75 & 2.20 & 2.10 & 2.00 & 4.40 & 4.00 & 4.30 & 0.55 & 1.40 & 1.30 & 1.25 & 3.60 & 3.20 & 3.60 \\
\hline
\end{longtable}

\begin{longtable}{|c||*{7}{c|}|*{7}{c|}}
\caption{RMSE $\times 100$ for the partial coherence estimates obtained using seven estimation methods (EM1–EM7) across five models (M1–M5). Results are presented for both five-node and ten-node networks under increasing time series lengths $T = 100, 200, 500, 1000$.} \label{Sim:Table_no_miss_pcoh}\\ 
\hline
\multirow{2}{*}{\textbf{Model}} 
& \multicolumn{7}{c|}{\textbf{T = 100}} 
& \multicolumn{7}{c|}{\textbf{T = 200}} \\
\cline{2-15}
& \textbf{EM1} & \textbf{EM2} & \textbf{EM3} & \textbf{EM4} & \textbf{EM5} & \textbf{EM6} & \textbf{EM7} 
& \textbf{EM1} & \textbf{EM2} & \textbf{EM3} & \textbf{EM4} & \textbf{EM5} & \textbf{EM6} & \textbf{EM7} \\
\hline
\endfirsthead

\multicolumn{15}{c}{\textit{(Table \ref{Sim:Table_no_miss_pcoh} continued - T=100, 200)}} \\
\hline
\multirow{2}{*}{\textbf{Model}} 
& \multicolumn{7}{c|}{\textbf{T = 100}} 
& \multicolumn{7}{c|}{\textbf{T = 200}} \\
\cline{2-15}
& \textbf{EM1} & \textbf{EM2} & \textbf{EM3} & \textbf{EM4} & \textbf{EM5} & \textbf{EM6} & \textbf{EM7} 
& \textbf{EM1} & \textbf{EM2} & \textbf{EM3} & \textbf{EM4} & \textbf{EM5} & \textbf{EM6} & \textbf{EM7} \\
\hline
\endhead

\multicolumn{15}{r}{\textit{(continued on next page)}} \\
\endfoot
\endlastfoot

\multicolumn{15}{c}{\textit{Five Nodes Network}} \\
\hline
\textbf{M1} & 0.80 & 5.16 & 4.56 & 3.76 & 8.00 & 5.96 & 10.21 & 0.58 & 2.85 & 2.58 & 2.19 & 5.80 & 4.45 & 7.08 \\
\textbf{M2} & 0.94 & 5.08 & 5.08 & 4.04 & 9.85 & 6.01 & 9.85 & 0.66 & 2.86 & 2.86 & 2.68 & 6.90 & 4.55 & 6.90 \\
\textbf{M3} & 0.98 & 4.95 & 4.95 & 4.95 & 9.88 & 9.88 & 9.88 & 0.67 & 2.70 & 2.70 & 2.70 & 6.78 & 6.78 & 6.78 \\
\textbf{M4} & 1.15 & 7.30 & 7.30 & 7.30 & 9.81 & 9.81 & 9.81 & 0.81 & 3.92 & 3.92 & 3.92 & 6.81 & 6.81 & 6.81 \\
\textbf{M5} & 1.43 & 7.39 & 7.39 & 7.39 & 9.90 & 9.90 & 9.90 & 0.90 & 3.85 & 3.85 & 3.85 & 6.80 & 6.80 & 6.80 \\
\hline
\multicolumn{15}{c}{\textit{Ten Nodes Network}} \\
\hline
\textbf{M1} & 0.43 & 5.52 & 3.94 & 2.99 & 6.21 & 4.05 & 14.97 & 0.32 & 2.95 & 2.21 & 1.73 & 4.54 & 3.08 & 9.25 \\
\textbf{M2} & 0.39 & 5.47 & 4.88 & 2.91 & 9.91 & 3.90 & 14.90 & 0.25 & 2.84 & 2.58 & 1.61 & 6.69 & 2.90 & 9.08 \\
\textbf{M3} & 0.42 & 5.46 & 5.46 & 2.96 & 14.92 & 3.94 & 14.92 & 0.28 & 2.82 & 2.82 & 1.56 & 9.02 & 2.89 & 9.02 \\
\textbf{M4} & 0.47 & 8.28 & 8.28 & 4.15 & 14.85 & 3.89 & 14.85 & 0.34 & 4.21 & 4.21 & 2.20 & 9.09 & 2.94 & 9.09 \\
\textbf{M5} & 0.54 & 8.32 & 8.32 & 4.23 & 14.85 & 3.94 & 14.85 & 0.35 & 4.16 & 4.16 & 2.20 & 9.09 & 2.96 & 9.09 \\
\hline
\end{longtable}

\vspace{-1.5em} 
\addtocounter{table}{-1} 

\begin{longtable}{|c||*{7}{c|}|*{7}{c|}}
\caption{\textit{(continued)}} \\ 
\hline
\multirow{2}{*}{\textbf{Model}} 
& \multicolumn{7}{c|}{\textbf{T = 500}} 
& \multicolumn{7}{c|}{\textbf{T = 1000}} \\
\cline{2-15}
& \textbf{EM1} & \textbf{EM2} & \textbf{EM3} & \textbf{EM4} & \textbf{EM5} & \textbf{EM6} & \textbf{EM7} 
& \textbf{EM1} & \textbf{EM2} & \textbf{EM3} & \textbf{EM4} & \textbf{EM5} & \textbf{EM6} & \textbf{EM7} \\
\hline
\endfirsthead

\multicolumn{15}{c}{\textit{(Table \ref{Sim:Table_no_miss_pcoh} continued - T=500, 1000)}} \\
\hline
\multirow{2}{*}{\textbf{Model}} 
& \multicolumn{7}{c|}{\textbf{T = 500}} 
& \multicolumn{7}{c|}{\textbf{T = 1000}} \\
\cline{2-15}
& \textbf{EM1} & \textbf{EM2} & \textbf{EM3} & \textbf{EM4} & \textbf{EM5} & \textbf{EM6} & \textbf{EM7} 
& \textbf{EM1} & \textbf{EM2} & \textbf{EM3} & \textbf{EM4} & \textbf{EM5} & \textbf{EM6} & \textbf{EM7} \\
\hline
\endhead

\multicolumn{15}{r}{\textit{(continued on next page)}} \\
\endfoot
\endlastfoot

\multicolumn{15}{c}{\textit{Five Nodes Network}} \\
\hline
\textbf{M1} & 0.37 & 1.42 & 1.32 & 1.17 & 3.78 & 3.00 & 4.43 & 0.27 & 0.91 & 0.87 & 0.81 & 2.78 & 2.26 & 3.16 \\
\textbf{M2} & 0.44 & 1.45 & 1.45 & 1.99 & 4.39 & 3.30 & 4.39 & 0.30 & 0.91 & 0.91 & 1.79 & 3.11 & 2.65 & 3.11 \\
\textbf{M3} & 0.42 & 1.31 & 1.31 & 1.31 & 4.26 & 4.26 & 4.26 & 0.29 & 0.79 & 0.79 & 0.79 & 3.10 & 2.67 & 3.10 \\
\textbf{M4} & 0.47 & 2.31 & 2.31 & 1.99 & 4.40 & 3.90 & 4.40 & 0.33 & 1.12 & 1.12 & 1.16 & 3.63 & 3.25 & 3.63 \\
\textbf{M5} & 0.54 & 2.48 & 2.48 & 2.02 & 4.55 & 3.98 & 4.55 & 0.34 & 1.23 & 1.23 & 1.17 & 3.77 & 3.35 & 3.77 \\
\hline
\multicolumn{15}{c}{\textit{Ten Nodes Network}} \\
\hline
\textbf{M1} & 0.44 & 1.39 & 1.35 & 1.17 & 4.01 & 3.24 & 5.01 & 0.31 & 0.98 & 0.94 & 0.83 & 2.91 & 2.44 & 3.49 \\
\textbf{M2} & 0.48 & 1.44 & 1.43 & 1.22 & 4.52 & 3.38 & 5.05 & 0.34 & 0.99 & 0.98 & 0.90 & 3.17 & 2.53 & 3.50 \\
\textbf{M3} & 0.47 & 1.43 & 1.43 & 1.18 & 4.58 & 3.46 & 5.07 & 0.32 & 0.97 & 0.97 & 0.90 & 3.25 & 2.62 & 3.55 \\
\textbf{M4} & 0.50 & 2.30 & 2.30 & 1.89 & 4.72 & 3.86 & 5.18 & 0.36 & 1.11 & 1.11 & 1.02 & 3.67 & 3.12 & 3.85 \\
\textbf{M5} & 0.56 & 2.45 & 2.45 & 2.02 & 4.79 & 3.94 & 5.22 & 0.38 & 1.18 & 1.18 & 1.04 & 3.75 & 3.19 & 3.93 \\
\hline
\end{longtable}

\subsection{Robustness of spectral quantities under structural misspecification}\label{sim:miss}

We now examine how uncertainty regarding the underlying network process affects the accuracy of frequency-domain recovery. In this scenario, the parametric spectral estimator must first identify the model orders (lag and neighborhood depth) using the Bayesian Information Criterion (BIC) \citep{BIC}, as described in Section 3 of \cite{Knight2020}. This test determines whether the spectral estimators are robust to the errors typically introduced during the model-fitting stage of a TS-NSD process.

RMSE results for the estimated spectrum, coherence, and partial coherence are summarized in Table~\ref{tab:rmse-miss}. While we observe a slight depreciation in spectral accuracy for small sample sizes—for instance, an RMSE increase by a factor of $1.25 \times 10^{-2}$ at $T=100$—reassuringly, the asymptotic performance is maintained. This indicates that even when the exact parametric form must be estimated from data, the resulting spectral characterization remains highly reliable as $T$ increases.

\begin{table}[!ht]
\centering
\caption{RMSE ($\times 100$) for spectral quantities under model misspecification, evaluating the impact of BIC-based order selection across different sample sizes.}
\label{tab:rmse-miss} 
\resizebox{\textwidth}{!}{%
\begin{tabular}{l|cccc|cccc|cccc}
\toprule
\textbf{Model} & \multicolumn{4}{c|}{\textbf{Spectrum (\%)}} & \multicolumn{4}{c|}{\textbf{Coherence (\%)}} & \multicolumn{4}{c}{\textbf{Partial Coherence (\%)}} \\
\midrule
& T=100 & T=200 & T=500 & T=1000 & T=100 & T=200 & T=500 & T=1000 & T=100 & T=200 & T=500 & T=1000 \\
\midrule
\multicolumn{13}{c}{\textit{Five Nodes Network}} \\
\midrule
\textbf{M1}   & 9.90 & 6.80 & 4.20 & 2.96 & 1.25 & 0.83 & 0.50 & 0.36 & 1.04 & 0.68 & 0.40 & 0.28 \\
\textbf{M2}   & 8.99 & 5.46 & 3.33 & 2.20 & 1.45 & 0.89 & 0.55 & 0.37 & 1.26 & 0.75 & 0.47 & 0.31 \\
\textbf{M3}   &11.85 & 8.85 & 5.17 & 3.28 & 1.52 & 1.13 & 0.68 & 0.44 & 1.26 & 0.91 & 0.53 & 0.33 \\
\textbf{M4}   &11.61 & 8.93 & 5.96 & 3.81 & 1.71 & 1.30 & 0.88 & 0.55 & 1.37 & 1.02 & 0.66 & 0.41 \\
\textbf{M5}   &14.11 &11.37 & 7.90 & 5.47 & 1.99 & 1.55 & 1.08 & 0.73 & 1.49 & 1.10 & 0.77 & 0.53 \\
\midrule
\multicolumn{13}{c}{\textit{Ten Nodes Network}} \\
\midrule
\textbf{M1}   & 5.63 & 3.85 & 2.38 & 1.58 & 0.72 & 0.49 & 0.29 & 0.19 & 0.63 & 0.43 & 0.25 & 0.17 \\
\textbf{M2}   & 4.97 & 2.95 & 1.80 & 1.23 & 0.63 & 0.35 & 0.23 & 0.16 & 0.55 & 0.30 & 0.19 & 0.13 \\
\textbf{M3}   & 6.68 & 4.65 & 2.66 & 1.74 & 0.74 & 0.49 & 0.29 & 0.19 & 0.61 & 0.38 & 0.21 & 0.14 \\
\textbf{M4}   & 6.12 & 4.72 & 2.88 & 1.76 & 0.73 & 0.56 & 0.34 & 0.22 & 0.59 & 0.43 & 0.26 & 0.16 \\
\textbf{M5}   & 8.14 & 6.59 & 4.31 & 2.56 & 0.97 & 0.77 & 0.50 & 0.29 & 0.67 & 0.49 & 0.31 & 0.17 \\
\bottomrule
\end{tabular}%
}
\end{table}

\subsection{Performance evaluation in nonparametric TS-NSD frameworks}

In the simulation studies presented in Sections~\ref{sim:no_miss}--\ref{sim:miss}, the parametric estimators consistently exhibited superior performance. This behavior is expected, as in those scenarios, the data-generating process (DGP) was explicitly derived from parametric modeling, providing the parametric estimators with the advantage of a correctly specified model. However, in real-world network time series applications, the true underlying DGP is typically unknown. To assess the robustness of our framework under model misspecification, we establish a nonparametric DGP where the stationarity and network dependence are defined through a wavelet-based spectral structure.

To generate these processes, we construct a spectral density directly in the wavelet domain. For a set of active scales $j$, we define an inverse spectrum (precision matrix) $\Theta_j$ that respects the network topology by setting $\Theta_{j; r,s} = -0.3$ for $A_{r,s} = 1$ and ensuring positive definiteness via diagonal dominance. The local spectrum is then given by $S_j = \Theta_j^{-1}$, which implies that the partial coherence is inherently linked to the network adjacency matrix. The innovations $\xi_{k}$ are sampled from a multivariate normal distribution $\mbox{N}(0, \Lambda_j)$, where $\Lambda_j$ is the correlation matrix derived from $S_j$. Finally, the total signal is reconstructed via an Average Basis synthesis \citep{wavethresh} across all active scales. The partial coherence results appear in Table~\ref{tab:nonparam_results}.

\begin{table}[ht]
\centering
\caption{Partial coherence RMSE (\(\times 100\)) for the nonparametric DGP across 500 simulations. Methods compared are: Parametric (\textbf{GNAR}), unconstrained VAR (\textbf{UVAR}), unconstrained nonparametric (\textbf{UNP}), constrained VAR (\textbf{CVAR}), and constrained nonparametric (\textbf{CNP}).}
\label{tab:nonparam_results}
\begin{tabular}{lccccc}
\hline
\textbf{T} & \textbf{GNAR} & \textbf{UVAR} & \textbf{UNP} & \textbf{CVAR} & \textbf{CNP} \\ \hline
$T=128$  & 11.8744 & 11.5317 & 12.0881 & 9.4467  & 7.6769 \\
$T=256$  & 12.0356 & 10.9579 & 8.3856  & 10.3978 & 5.8492 \\
$T=512$  & 12.1336 & 11.3135 & 6.5961  & 11.1724 & 4.9080 \\
$T=1024$ & 12.1710 & 11.7026 & 5.1026  & 11.6626 & 3.9741 \\ \hline
\end{tabular}
\end{table}

Table~\ref{tab:nonparam_results} illustrates a significant shift in performance compared to parametric-led simulations. Under this nonparametric DGP, both the GNAR and VAR-based estimators show stagnant RMSE values regardless of the increase in sample size $T$, indicating a fundamental lack of adaptability to the model misspecification. In contrast, the nonparametric estimators demonstrate consistent convergence, with the RMSE decreasing as $T$ grows. Notably, the network-penalized nonparametric estimator (\textbf{CNP}) achieves the best performance across all sample sizes, confirming that incorporating network information into a flexible nonparametric framework provides the most robust path for spectral estimation when the true generating process is unknown.

\section{Global bank network connectedness and its spectral features}\label{sec:application}

Using our network-informed spectral estimators, we next investigate frequency-specific patterns of connectedness among financial institutions across different countries and regions. For our motivating example, our proposal provides the formal methodology to resolve the frequency-specific composition of network interactions, enabling the decomposition of complex spillovers into their short- and long-term components, a dimension of connectivity that time-domain models/ standard GSP tools inherently aggregate and obscure.


Understanding the interdependence among financial institutions is essential for assessing systemic risk and the propagation of financial contagion. A widely used approach to capture these interconnections is through \emph{connectedness networks}, which quantify how shocks or uncertainty in one institution influence others, either directly or indirectly. This framework forms the basis of several influential studies; see, e.g., \citet{DieboldYilmaz2014,Barunik2018}.

Following \citet{Demirer2018}, we use data collected from 96 publicly traded banks across 29 economies from September 12, 2003 to February 7, 2014, with the sample including nearly all globally systemically important banks (GSIBs). Daily volatility for each bank \( i \) is estimated via the range-based Garman--Klass estimator \citep{GarmanKlass1980}
\begin{equation*}
\hat{\sigma}^2_{i,t} = 0.511(H_{i,t} - L_{i,t})^2 - 0.019\left[(C_{i,t} - O_{i,t})(H_{i,t} + L_{i,t} - 2O_{i,t}) - 2(H_{i,t} - O_{i,t})(L_{i,t} - O_{i,t})\right] - 0.383(C_{i,t} - O_{i,t})^2,
\end{equation*}
where \( O_{i,t}, H_{i,t}, L_{i,t}, C_{i,t} \) denote the logs of daily high, low, open, and close prices of bank \( i \) on day \( t \).


As our framework requires weak stationarity, we apply the locally stationary wavelet (LSW) test \citep{Nason2013} to rigorously assess second-order stationarity, addressing limitations of standard unit root tests which often overlook time-varying dependencies in the covariance structure \citep{Munoz2024}. After a logarithmic transformation to stabilize variance, 57 of the original 96 banks pass the stationarity test at the 5\% significance level. These banks (and corresponding countries) are listed in Table~\ref{app:table_banks} in Appendix~\ref{app:D}. 


We proceed with the dataset $(\mathbf{X}_t, G)$ where $X_{i,t} = \log(\hat{\sigma}_{i,t})$ denotes the process at node $i$ of the graph $G=(K,E)$, constructed using volatility spillovers by means of the Generalized Forecast Error Variance Decomposition (GFEVD) \citep{Pesaran1998}. The resulting topology, obtained following the next three steps, will provide the infrastructure to infer how multi-hop paths mediate frequency-domain dependencies. 

\paragraph{Step 1: Sparse VAR estimation.}
Letting \( \mathbf{X}_t \in \mathbb{R}^d \) be the log-volatility series across \( d = 57 \) banks, we estimate a sparse VAR($p$) model: $\mathbf{X}_t = \sum_{k=1}^{p} \mathbf{\Pi}_k \mathbf{X}_{t-k} + \bm{\varepsilon}_t$. We employ Lasso regression with lag order $p$ selected by the BIC \citep{BIC} and 10-fold cross-validation to select the penalty parameter.

\paragraph{Step 2: Moving average (MA) representation.}
The estimated coefficients are used to obtain the MA representation $\mathbf{X}_t = \sum_{h=0}^{\infty} \mathbf{B}_h \bm{\varepsilon}_{t-h}$ up to a finite forecast horizon \( H = 10 \).

\paragraph{Step 3: GFEVD and network construction.}
The GFEVD matrix \( \Psi(H) = [\psi_{ij}(H)] \) is computed as
\begin{equation*}
\psi_{ij}(H) = \frac{1}{\widehat{\mathbf{V}}_{jj}}  \times \frac{\sum_{h=1}^{H}\left( \mathbf{e}_i^\top \mathbf{B}_h \widehat{\mathbf{V}} \mathbf{e}_j \right)^2}{\sum_{h=1}^{H} \mathbf{e}_i^\top \mathbf{B}_h \widehat{\mathbf{V}}\mathbf{B}_h^\top \mathbf{e}_i}.
\end{equation*}
The network is row-normalized and thresholded at $\tau^*$ to ensure connectivity. Edges in the final undirected network $G=(K,E)$ are weighted by the symmetric average $w_{ij} = \frac{1}{2} ( \psi_{ij} + \psi_{ji} )$.

\paragraph{Resulting empirical network.}
Figure~\ref{fig:banking_network} shows the estimated undirected network; consistent with \citet{Demirer2018}, the network exhibits clear regional clustering. Due to the separation of Asian countries (green), the full layout in Figure~\ref{fig:banking_network} appears visually cluttered, hence for clarity Figure~\ref{fig:network2} in Appendix~\ref{app:D} provides a zoomed-in view excluding Asian countries, while Figure~\ref{fig:Stg_net} depicts the \( r \)-stage neighborhood structures, revealing dependencies extending up to five stages.

To analyze the frequency-domain inter-bank connectedness, we model $(\mathbf{X}_t, G)$ using the framework in Definition~\ref{def:GNAR_Spec} that naturally captures autoregressive dynamics subject to a static network structure, and allows the use of {\em higher-order} adjacency matrices, thus enabling inference based on wider spatial connectivity coupled with nodal (bank) temporal dynamics. We use the PNACF Corbit plot of \cite{nason2025} (Figure~\ref{fig:corbit}, Appendix~\ref{app:D}) to select a GNAR\((2, [1, 1])\) model. The resulting parametric spectral density matrix \(\widehat{\mathbf{f}}(\cdotp,G)\) allows us to disentangle frequency responses while maintaining structural parsimony, as next described.

To examine the impact of connection strength locality across variable time-spans (as captured by their corresponding frequencies), we focus on nine representative node pairs: 
three within-country connected ($r=1$) (\textbf{JPMorgan Chase \& Co.}--\textbf{Bank of America} \texttt{jpm.us}--\texttt{bac.us} , \textbf{BNP Paribas}--\textbf{Credit Agricole} \texttt{bnp.fr}--\texttt{aca.fr}, and \textbf{Mizuho Financial Group}--\textbf{Resona Holdings} \texttt{mzh.jp}--\texttt{rsnh.jp}); 
three cross-country connected ($r=1$) (\textbf{JPMorgan Chase \& Co.}--\textbf{Standard Chartered} \texttt{jpm.us}--\texttt{stan.gb}, \textbf{BNP Paribas}--\textbf{Banco Santander} \texttt{bnp.fr}--\textbf{san.es}, and \textbf{Mizuho Financial Group}--\textbf{Woori Finance Holdings} \texttt{mzh.jp}--\texttt{wrfh.kr}); 
and three disconnected ($r>2$) linked only through multi-stage paths (\textbf{JPMorgan Chase \& Co.}--\textbf{Resona Holdings} \texttt{jpm.us}--\texttt{rsnh.jp}, \textbf{China Merchants Bank}--\textbf{Mediobanca Banca di Credito Finanziario} \texttt{cmb.cn}--\texttt{mb.it}, and \textbf{China Merchants Bank}--\textbf{Shizuoka Bank} \texttt{cmb.cn}--\texttt{shzb.jp}). 

Figure~\ref{fig:gnar_selected_spectral_plots} displays several spectral summaries (modulus, phase, squared coherence, and squared partial coherence), illustrating the inter-bank differential volatility behaviour derived from their network connectivity. While similar patterns of behaviour are associated with lower inter-bank network distances, the strength of volatility interconnectednes is markedly different across the chosen pairs, with within-country pairs exhibiting the strongest volatility coupling, which in some cases are notably almost matched by well-connected cross-country pairs. Regardless of the short-term coupling intensity, all pairs exhibit decay with increasing time-spans (frequencies). 

Spatially, Figure~\ref{fig:gnar_selected2} in Appendix~\ref{app:D} examines the 
disconnected node pairs corresponding to $r=3,4,5-$stage neighbors. 
The results indicate a clear spatial decay in volatility co-movements; 
inter-bank dependencies weaken as the network distance $(r)$ increases, 
a pattern that remains consistent across all frequencies. Phase behavior provides further insight: while adjacent ($r=1$ in Figure~\ref{fig:gnar_selected_spectral_plots}) nodal pairs exhibit near-zero phase (synchronous co-movement), disconnected pairs in Figure~\ref{fig:gnar_selected2} display significant variation at low frequencies. Notably, the 5-stage neighbor pair exhibits a consistent phase inversion, indicating opposing volatility behavior across the spectrum.

\begin{figure}[!ht]
    \centering
    \includegraphics[width=0.45\textwidth]{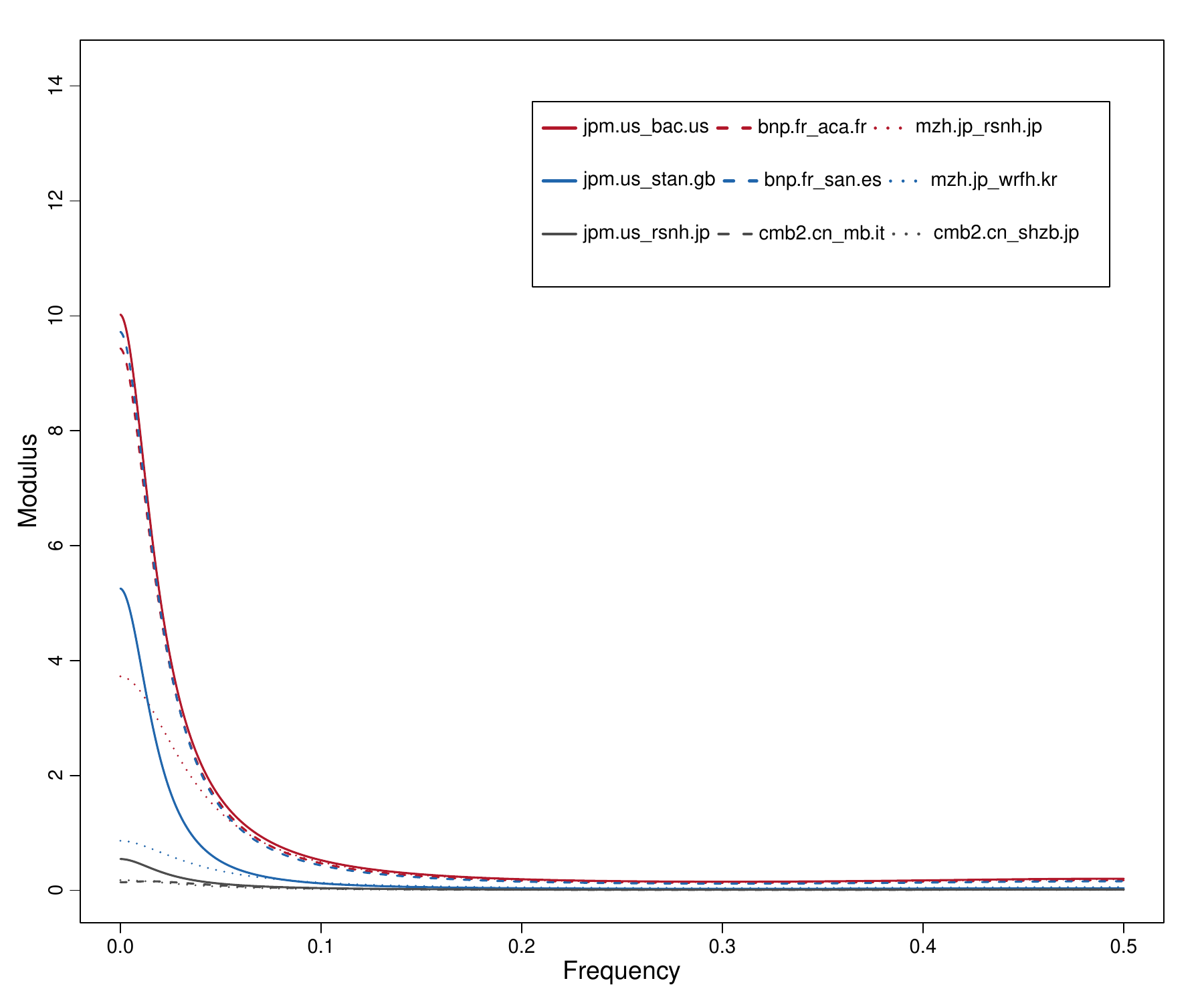} \quad
    \includegraphics[width=0.45\textwidth]{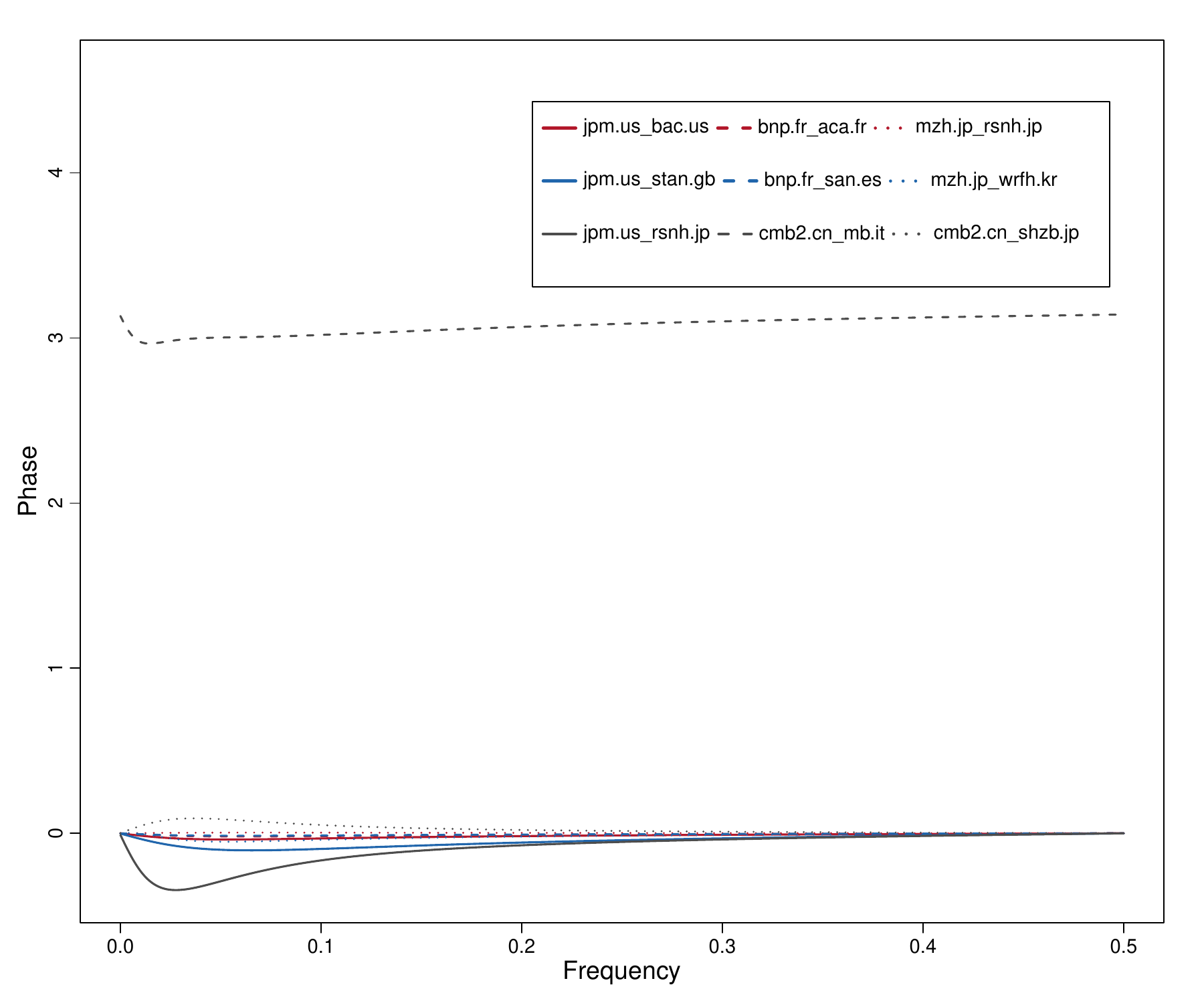} \\
    \vspace{0.5em}
    \includegraphics[width=0.45\textwidth]{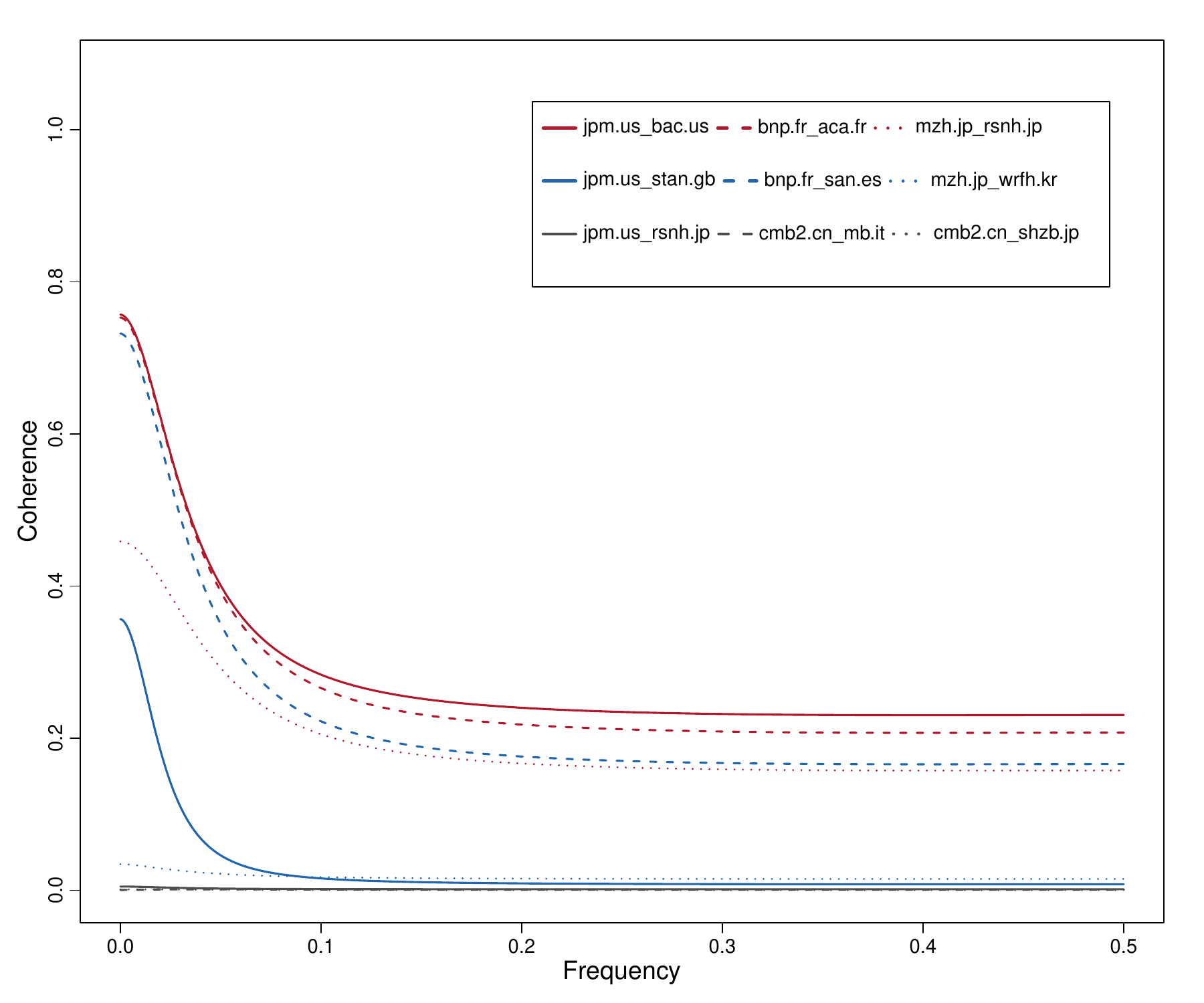} \quad
    \includegraphics[width=0.45\textwidth]{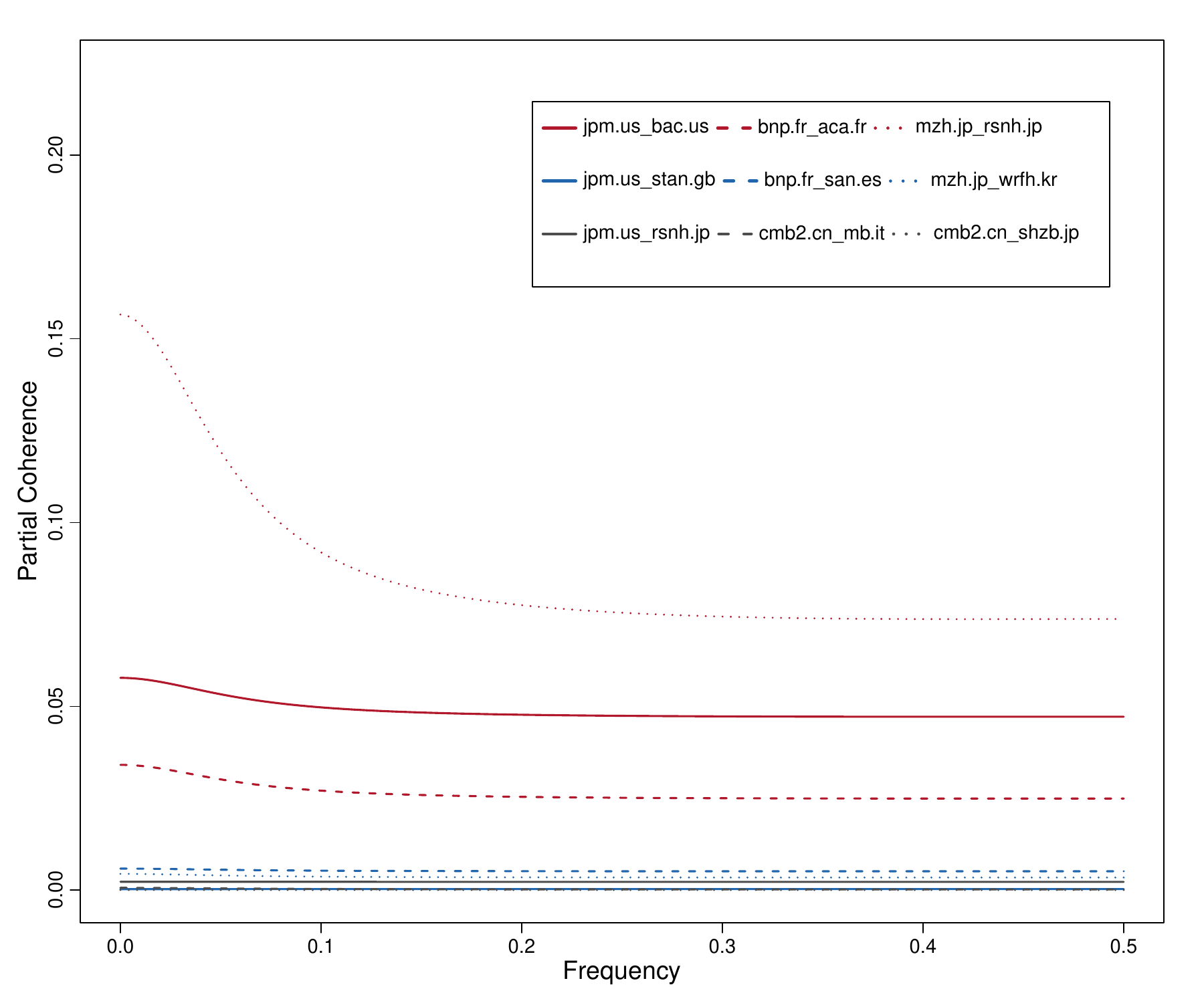}
    \caption{Spectral summaries for nine selected bank pairs $(i,j)$. Top row: modulus \(|[\widehat{\mathbf{f}}(\cdotp,G)]_{ij}|\) and phase \(\arg([\widehat{\mathbf{f}}(\cdotp,G)]_{ij})\) of estimated cross-spectra. Bottom row: squared coherence \([\widehat{\boldsymbol{\rho}}(\cdotp,G)]_{ij}^2\) and partial coherence \([\widehat{\boldsymbol{\gamma}}(\cdotp,G)]_{ij}^2\). Colors indicate connection type: red (same-country), blue (cross-country), gray (disconnected).}
    \label{fig:gnar_selected_spectral_plots}
\end{figure}




\section{Discussion}
We developed an integrated spectral estimation framework for time series with network structured dependence (TS-NSD) that explicitly incorporates the known network topology into the frequency-domain analysis, via both parametric and nonparametric approaches. By including multi-hop dependencies beyond immediate neighbors, our proposal captures a more nuanced cross-nodal signal propagation across different Fourier frequencies. This integration of graph topology into spectral estimation offers a flexible and interpretable tool for analyzing complex network time series data, and has wide applicability across sectors where network data is collected. The effectiveness of our methods was studied through extensive simulations, showcasing setups where each of the parametric and nonparametric estimators perform strongly. The practical utility of our proposal was demonstrated through a global bank network connectedness analysis, where it revealed frequency-dependent inter-bank volatility (synchronous or anti-synchronous) co-movement relationships closely connected to their network distances. These results underscore the framework's potential to improve understanding of systemic risk by revealing how inter-bank dependencies evolve over short versus long time-spans, and across network structures in global financial systems.
Whilst this study assumes a fixed network topology, the proposal provides a robust foundation for extensions into dynamic settings, such as incorporating nonstationary time series over nodes and edges, or treating the network as a dynamic stochastic process \citep{Krampe2019} to better capture evolving systemic risks.


\section*{Supplementary Materials}

\textbf{Code Reproducibility.} The R code used to reproduce the simulation studies described in this paper is publicly available at the following repository:  \href{https://github.com/cfjimenezv07/GNAR_Spectral_Analysis}{Github repository}.


\section*{Funding}
The authors gratefully acknowledge support from the EPSRC NeST Programme Grant EP/X002195/1.

\bibliographystyle{chicago}
\bibliography{York}

\appendix
\renewcommand{\thetable}{\Alph{section}.\arabic{table}}
\renewcommand{\thefigure}{\Alph{section}.\arabic{figure}}
\setcounter{table}{0}
\setcounter{figure}{0}

\section{Construction of the nonparametric estimator in Section~\ref{GNAR_spec_PN}}\label{app:1} 

Recall we defined the smoothed periodogram matrix using a bandwidth $m_T=o(\sqrt T)$ and kernel $W_T(\cdotp)$ as 

\begin{equation}\nonumber
\widetilde{\mathbf{I}}_T(\omega_l) = \frac{1}{2 m_T + 1} \sum_{|k| \leq m_T} W_T(k) \mathbf{I}_T\left(\omega_{l+k}\right), 
\end{equation}
with the indices \({l+k}\) understood as evaluated modulo $T$.
For completeness, the Fourier periodogram matrix used above is defined at frequency \(\omega_l=l/T, \, l=0, \, 1, \ldots n_T=([T/2]-1)\), as in Section~\ref{sec:Np_AN}~\ref{rawfourier},
\[
\mathbf{I}_T(\omega_l) = \mathbf{J}(\omega_l) \overline{\mathbf{J}}(\omega_l)^{\top},
\]
with the discrete Fourier transform (DFT) defined as
\[
\mathbf{J}(\omega_l) = \frac{1}{\sqrt{T}} \sum_{t=1}^T \mathbf{X}_t e^{-i 2 \pi t \omega_l}.
\]

We split the DFT into its real and imaginary parts as
\[
\mathbf{J}(\omega_l) = \mathbf{R}(\omega_l) - i \mathbf{B}(\omega_l),
\]
where
\[
\mathbf{R}(\omega_l) := \operatorname{Re}(\mathbf{J}(\omega_l)) = \frac{1}{\sqrt{T}} \sum_{t=1}^T \mathbf{X}_t \cos(2 \pi t \omega_l),
\]
and
\[
\mathbf{B}(\omega_l) := \operatorname{Im}(\mathbf{J}(\omega_l)) = -\frac{1}{\sqrt{T}} \sum_{t=1}^T \mathbf{X}_t \sin(2 \pi t \omega_l).
\]
Expressed under real-valued form, $\left(\mathbf{R}^\top(\omega_l), \,  \mathbf{B}^\top(\omega_l) \right) ^\top \dot\sim N_{2d}\left(\mathbf{0}, {\mathbf{\Sigma}}(\omega_l)\right)$ \citep{Shumway2017}, where $\mathbf{f}(\omega_l)={\mathbf C}(\omega_l)- i {\mathbf Q}(\omega_l)$ and its corresponding $2d \times 2d$ covariance matrix at frequency $\omega_l$ is 
\begin{equation}\label{eq:truecov}
{\mathbf{\Sigma}}(\omega_l) = {\frac{1}{2}}
\begin{bmatrix}
{\mathbf C}(\omega_l) & -{\mathbf Q}(\omega_l) \\
{\mathbf Q}(\omega_l) & {\mathbf C}(\omega_l)
\end{bmatrix}.
\end{equation}
Substituting the real and imaginary DFT parts into the periodogram and using the Hermitian transpose, 
\[
\mathbf{I}_T(\omega_l) = (\mathbf{R}(\omega_l) - i \mathbf{B}(\omega_l))(\mathbf{R}^\top(\omega_l) + i \mathbf{B}^\top (\omega_l)) = \mathbf{R} (\omega_l)\mathbf{R}^\top (\omega_l)+ \mathbf{B}(\omega_l) \mathbf{B}^\top (\omega_l) - i \left( \mathbf{B} (\omega_l)\mathbf{R}^\top (\omega_l) - \mathbf{R} (\omega_l) \mathbf{B}^\top (\omega_l) \right).
\]
Thus, the real part of the periodogram matrix is
\[
\operatorname{Re}(\mathbf{I}_T(\omega_l)) = \mathbf{R}(\omega_l) \mathbf{R}^\top(\omega_l) + \mathbf{B}(\omega_l) \mathbf{B}^\top(\omega_l),
\]
and the imaginary part is
\[
\operatorname{Im}(\mathbf{I}_T(\omega_l)) =  \mathbf{B}(\omega_l) \mathbf{R}^\top(\omega_l)  - \mathbf{R}(\omega_l) \mathbf{B}^\top(\omega_l).
\]
Finally, the smoothed periodogram matrix can be expressed as
\[
\widetilde{\mathbf{I}}_T(\omega_l) = \widetilde{\mathbf{C}}(\omega_l) - i \widetilde{\mathbf{Q}}(\omega_l) ,
\]
where the real and imaginary parts are given respectively by
\[
\widetilde{\mathbf{C}}(\omega_l) = \frac{1}{2 m_T+ 1} \sum_{|k| \leq m_T} W_T(k) \left[ \mathbf{R}(\omega_{l+k}) \mathbf{R}^\top(\omega_{l+k}) + \mathbf{B}(\omega_{l+k}) \mathbf{B}^\top(\omega_{l+k}) \right],
\]
and
\[
\widetilde{\mathbf{Q}}(\omega_l) = \frac{1}{2 m_T+ 1} \sum_{|k| \leq m_T} W_T(k) \left[ \mathbf{B}(\omega_{l+k}) \mathbf{R}^\top(\omega_{l+k}) - \mathbf{R}(\omega_{l+k}) \mathbf{B}^\top(\omega_{l+k}) \right].
\]
Note that $\widetilde{\mathbf{C}}(\omega_l)=\widetilde{\mathbf{C}}^{\top}(\omega_l)$ and $\widetilde{\mathbf{Q}}(\omega_l)=-\widetilde{\mathbf{Q}}^{\top}(\omega_l)$ for any frequency $\omega_l$. Then, in the same vein as \cite[Theorem C.6]{Shumway2017}, the smoothed periodogram matrix \(\widetilde{\mathbf{I}}_T(\omega_l)\) can be equivalently represented in real-valued form as the symmetrical matrix 
\[
\widetilde{\mathbf{\Sigma}}(\omega_l) = {\frac{1}{2}}
\begin{bmatrix}
\widetilde{\mathbf C}(\omega_l) & -\widetilde{\mathbf Q}(\omega_l) \\
\widetilde{\mathbf Q}(\omega_l) & \widetilde{\mathbf C}(\omega_l)
\end{bmatrix},
\]
which is a well-behaved estimator for the true covariance matrix ${\mathbf{\Sigma}}(\omega_l)$ defined in equation~\eqref{eq:truecov}.

\section{Asymptotic properties of the nonparametric estimator in Section~\ref{GNAR_spec_PN}}\label{app:B}

\renewcommand{\theprop}{\Alph{section}.\arabic{prop}}

To study the asymptotic properties of the nonparametric spectral estimator introduced in Section~\ref{GNAR_spec_PN}, we first establish some theoretical results adapted to this context, drawing in particular on \cite{Dempster1972} and Chapter 5 of \cite{lauritzen1996}.
 We establish the notation we use from here on in-line with graphical models. Stemming from the original graph $G=(K,E)$ with $d$ nodes in set $K$ and edges contained in the set $E$, we define an new undirected graph $\mathcal{G} = (\Gamma, \mathcal E)$ over $2d$ nodes, where $\Gamma = \{1, \ldots, 2d\}$ is the set of (new) nodes, and the new set of edges $\mathcal E \subseteq \Gamma \times \Gamma$ is encoded by the augmented adjacency matrix $\tilde{\mathbf{A}}_1 \in \{0,1\}^{2d \times 2d}$ generated from the original graph $G$ as described in equation~\eqref{eq:A1}.

\begin{proof}[Proof of Proposition~\ref{prop:Nonpar_Spec_Est}]
In order to establish the desired consistency result for the constrained spectral estimator, let us first establish the following two results.
\begin{prop}\citep[Adapted from Proposition 5.2 in][]{lauritzen1996}\label{Prop_cond_indep}
    Assume that $\widetilde{F}(\omega_l) \in \mathbb{R}^{2d}$ is a random vector such that $\widetilde{F}(\omega_l) \sim \mathcal{N}(0, \mathbf{\Sigma}(\omega_l))$, and let $\mathbf{\Theta}(\omega_l) = \mathbf{\Sigma}^{-1}(\omega_l) \in \mathbb{R}^{2d \times 2d}$ be the corresponding precision matrix. Then the pairwise conditional independence statement
\[
\widetilde{F}_\gamma(\omega_l) \perp \widetilde{F}_\mu(\omega_l) \mid \widetilde{F}_{\Gamma \setminus (\gamma, \mu)}(\omega_l) \iff \mathbf{\Theta}_{\gamma, \mu}(\omega_l) = 0
\]
holds if and only if $(\gamma, \mu) \notin \mathcal E$, i.e., the absence of an edge in the graph $\mathcal{G}$ implies conditional independence between the corresponding components of $\widetilde{F}(\omega_l)$ given all other components. 
\end{prop}

\begin{proof}
   Proposition~\ref{Prop_cond_indep} adapts Proposition~5.2 in \cite{lauritzen1996}, which presents the result for the multivariate normal distribution. Full details and proofs can be found therein.
\end{proof}

\begin{prop}
\label{thm:MLE}
\begin{enumerate}
    \item \label{item:existence} 
    Let \(\widetilde{\boldsymbol{\Sigma}}(\omega_l)\) be a well-behaved estimator of the covariance matrix \({\boldsymbol{\Sigma}}(\omega_l)\). Then there exists a unique estimator \(\widehat{\widetilde{\boldsymbol{\Sigma}}}(\omega_l)\) of \(\boldsymbol{\Sigma}(\omega_l)\) satisfying
    \[
    \widehat{\widetilde{\boldsymbol{\Sigma}}}_{\gamma, \mu}(\omega_l) =
    \widetilde{\boldsymbol{\Sigma}}_{\gamma, \mu}(\omega_l), \quad \text{for all } (\gamma,\mu) \in \mathcal{E} \text{ or } \gamma=\mu,
    \]
    and the constraint
    \[
    \widehat{\widetilde{\boldsymbol{\Theta}}}_{\gamma, \mu}(\omega_l) = \widehat{\widetilde{\boldsymbol{\Sigma}}}^{-1}_{\gamma, \mu}(\omega_l) = 0, \quad \text{for all } (\gamma, \mu) \notin \mathcal{E}.
    \]
    
    \item \label{item:MLE} 
    Among all multivariate normal models satisfying the constraints in item~\ref{item:existence}, the matrix \(\widehat{\widetilde{\boldsymbol{\Sigma}}}(\omega_l)\) is the maximum likelihood estimator (MLE) of \(\boldsymbol{\Sigma}(\omega_l)\).
\end{enumerate}
\end{prop}

\begin{proof}
\begin{enumerate}
 \item 
 Let us first re-express the constrained log-likelihood from Section~\ref{GNAR_spec_PN} over the set of edges \(\mathcal E\),
\begin{equation}
\label{eq:log_likelihood_graph_explicit}
\ell(\mathbf{\Theta}) = \log \det \mathbf{\Theta} - \operatorname{tr}(\widetilde{\mathbf{\Sigma}} \mathbf{\Theta}) - \sum_{\{(\gamma,\mu)\,|\,(\gamma,\mu)\notin \mathcal{E} \}} \lambda_{\gamma,\mu} \theta_{\gamma,\mu},
\end{equation}
where \(\theta_{\gamma,\mu}:= \mathbf{\Theta}_{\gamma, \mu}\) and \(\lambda_{\gamma,\mu}\) are Lagrange multipliers imposing the constraint that \(\theta_{\gamma,\mu} = 0\) for all \((\gamma,\mu) \notin \mathcal{E}\), and we note that we dropped the implicit frequency dependence.

Differentiating the log-likelihood in equation~\eqref{eq:log_likelihood_graph_explicit} with respect to each free parameter \(\theta_{\gamma,\mu}\), we obtain the following.
\noindent{\em For diagonal entries} \(\gamma = \mu\), we have
\begin{equation}
\frac{\partial \ell}{\partial \theta_{\gamma,\gamma}} 
= \frac{1}{2} \frac{\partial \log \det \boldsymbol{\Theta}}{\partial \theta_{\gamma,\mu}}
- \frac{1}{2} \widetilde{\boldsymbol{\Sigma}}_{\gamma, \gamma},
\label{B:eq:1}
\end{equation}
since no penalty is imposed on the diagonal.

\noindent{\em For off-diagonal entries} \(\gamma \neq \mu\), we distinguish two cases:

\begin{itemize}
\item If \((\gamma,\mu) \in \mathcal{E}\):
\begin{equation}
\frac{\partial \ell}{\partial \theta_{\gamma,\mu}} 
= \frac{1}{2} \frac{\partial \log \det \boldsymbol{\Theta}}{\partial \theta_{\gamma,\mu}}
- \widetilde{\boldsymbol{\Sigma}}_{\gamma, \mu}.
\label{B:eq:2}
\end{equation}

\item If \((\gamma,\mu) \notin \mathcal{E}\):
\begin{equation}
\frac{\partial \ell}{\partial \theta_{\gamma,\mu}} 
= \frac{1}{2} \frac{\partial \log \det \boldsymbol{\Theta}}{\partial \theta_{\gamma,\mu}}
- \widetilde{\boldsymbol{\Sigma}}_{\gamma, \mu}
- \lambda_{\gamma,\mu}.
\label{B:eq:2-penalty}
\end{equation}
\end{itemize}

Recall that the derivative of the log-determinant \citep[see page 641][]{boyd2004} is
\begin{equation*}
\frac{\partial \log \det \boldsymbol{\Theta}}{\partial \theta_{\gamma,\mu}}
= \operatorname{tr}\left( \boldsymbol{\Theta}^{-1} \frac{\partial \boldsymbol{\Theta}}{\partial \theta_{\gamma,\mu}} \right).
\end{equation*}

Since $\frac{\partial \boldsymbol{\Theta}}{\partial \theta_{\gamma,\mu}}$
is a matrix with a 1 in the \((\gamma,\mu)\) and \((\mu,\gamma)\) positions for \(\gamma \neq \mu\), and a 1 in \((\gamma,\gamma)\) for the diagonal case, the derivative simplifies to
\begin{eqnarray*}
\frac{\partial \log \det \boldsymbol{\Theta}}{\partial \theta_{\gamma,\gamma}} &=& \big( \boldsymbol{\Theta}^{-1}\big)_{\gamma,\gamma}. \mbox{ for the diagonal},\\
\frac{\partial \log \det \boldsymbol{\Theta}}{\partial \theta_{\gamma,\mu}} &=& 2 \big( \boldsymbol{\Theta}^{-1}\big)_{\gamma,\mu},
\mbox{ for the off-diagonal}.
\end{eqnarray*}

\medskip

Hence,  equations \eqref{B:eq:1}, \eqref{B:eq:2}, and \eqref{B:eq:2-penalty} become

\noindent{For} \(\gamma = \mu\):
\begin{equation}
\frac{1}{2} \big( \boldsymbol{\Theta}^{-1}(\omega_l) \big)_{\gamma,\gamma} - \frac{1}{2} \widetilde{\boldsymbol{\Sigma}}_{\gamma, \gamma}(\omega_l) = 0,
\quad \text{or} \quad
\big( \boldsymbol{\Theta}^{-1}(\omega_l) \big)_{\gamma,\gamma} = \widetilde{\boldsymbol{\Sigma}}_{\gamma, \gamma}(\omega_l).
\label{eq:est_eq_1}
\end{equation}

\noindent{For} \(\gamma \neq \mu\), with \((\gamma,\mu) \in \mathcal E\):
\begin{equation}
\big( \boldsymbol{\Theta}^{-1}(\omega_l) \big)_{\gamma,\mu} = \widetilde{\boldsymbol{\Sigma}}_{\gamma, \mu}(\omega_l).
\label{eq:est_eq_2}
\end{equation}

\noindent{For} \(\gamma \neq \mu\), with \((\gamma,\mu) \notin \mathcal E\), the corresponding constraint forces:
\begin{equation}
\theta_{\gamma,\mu} = 0.
\label{eq:constraint}
\end{equation}

Equations~\eqref{eq:est_eq_1}, \eqref{eq:est_eq_2}, and \eqref{eq:constraint} are precisely the conditions stated in~\ref{item:existence}.

\item The specific choice of \(\widehat{\widetilde{\boldsymbol{\Sigma}}}(\omega_l)\) that satisfies the conditions in \ref{item:existence} corresponds to the maximum likelihood estimator \(\boldsymbol{\Sigma}(\omega_l)\); see, for example, \citet[][Appendix A]{Dempster1972} and \citet[][Theorem 5.3]{lauritzen1996}.

\end{enumerate}
\end{proof}
Now by Proposition~\ref{thm:MLE}, \(\widehat{\widetilde{\boldsymbol{\Sigma}}}(\omega_l)\) is the maximum likelihood estimator of \(\boldsymbol{\Sigma}(\omega_l)\), hence the estimator satisfies the following asymptotic properties for any Fourier frequency $\omega_l$: 

\begin{enumerate}[(i)]
    \item {consistency of the precision matrix estimator:}
    $
    \widehat{\widetilde{\boldsymbol{\Theta}}}(\omega_l) \xrightarrow{p} \boldsymbol{\Theta}(\omega_l).
   $
    
    \item {consistency of the covariance matrix estimator:}
    $
    \widehat{\widetilde{\boldsymbol{\Sigma}}}(\omega_l) := \widehat{\widetilde{\boldsymbol{\Theta}}}^{-1}(\omega_l) \xrightarrow{p} \boldsymbol{\Sigma}(\omega_l).
    $
\end{enumerate}
Connecting (ii) above to the spectral estimator obtained as proposed in~\eqref{eq:GNAR_NP_PN} via the real-valued matrix in~\eqref{eq:tildeSigma} leads to the desired consistency result, $\widehat{\tilde{\mathbf{f}}}(\omega_l,G) \xrightarrow{p} \mathbf{f}(\omega_l,G)$.
\end{proof}

\begin{remark}
The results stated in Proposition~\ref{thm:MLE}, which are based on the graph pattern induced by the adjacency matrix \(\widetilde{\mathbf{A}}_1\), also apply to higher-order adjacency structure introduced in Section~\ref{sec:GNAR_induced}. As the order \(r^*\) increases, the number of free parameters grows accordingly, and the conditions for the validity of the maximum likelihood estimator continue to hold.
\end{remark}

\section{Simulation results for the estimated spectrum}\label{app:C}

Table~\ref{Sim:Table_no_miss_spectrum} presents the simulation results for Models (M1–M5) and estimation methods (EM1–EM7) described in Section~\ref{sec:sim_results}.

\setlength{\tabcolsep}{3pt}  
\renewcommand{\arraystretch}{0.9} 

\begin{longtable}{|c||*{7}{c|}|*{7}{c|}}
\caption{RMSE $\times 100$ for the spectrum estimates obtained using seven estimation methods (EM1–EM7) across five models (M1–M5). Results are presented for both five-node and ten-node networks under increasing time series lengths $T = 100, 200, 500, 1000$.} \label{Sim:Table_no_miss_spectrum}\\ 
\hline
\multirow{2}{*}{\textbf{Model}} 
& \multicolumn{7}{c|}{\textbf{T = 100}} 
& \multicolumn{7}{c|}{\textbf{T = 200}} \\
\cline{2-15}
& \textbf{EM1} & \textbf{EM2} & \textbf{EM3} & \textbf{EM4} & \textbf{EM5} & \textbf{EM6} & \textbf{EM7} 
& \textbf{EM1} & \textbf{EM2} & \textbf{EM3} & \textbf{EM4} & \textbf{EM5} & \textbf{EM6} & \textbf{EM7} \\
\hline
\endfirsthead

\multicolumn{15}{c}{\textit{(Table \ref{Sim:Table_no_miss_spectrum} continued - T=100, 200)}} \\
\hline
\multirow{2}{*}{\textbf{Model}} 
& \multicolumn{7}{c|}{\textbf{T = 100}} 
& \multicolumn{7}{c|}{\textbf{T = 200}} \\
\cline{2-15}
& \textbf{EM1} & \textbf{EM2} & \textbf{EM3} & \textbf{EM4} & \textbf{EM5} & \textbf{EM6} & \textbf{EM7} 
& \textbf{EM1} & \textbf{EM2} & \textbf{EM3} & \textbf{EM4} & \textbf{EM5} & \textbf{EM6} & \textbf{EM7} \\
\hline
\endhead

\multicolumn{15}{r}{\textit{(continued on next page)}} \\
\endfoot
\endlastfoot

\multicolumn{15}{c}{\textit{Five Nodes Network}} \\
\hline
\textbf{M1} & 7.88 & 32.23 & 31.40 & 27.44 & 30.85 & 27.53 & 31.76 & 5.97 & 21.59 & 21.99 & 19.46 & 26.58 & 23.63 & 27.00 \\
\textbf{M2} & 6.87 & 26.23 & 26.23 & 22.80 & 26.16 & 23.08 & 26.16 & 4.77 & 16.92 & 16.92 & 16.37 & 22.27 & 20.20 & 22.27 \\
\textbf{M3} & 8.91 & 27.39 & 27.39 & 27.39 & 27.11 & 27.11 & 27.11 & 6.27 & 17.77 & 17.77 & 17.77 & 23.23 & 23.23 & 23.23 \\
\textbf{M4} & 9.91 & 35.33 & 35.33 & 35.33 & 26.78 & 26.78 & 26.78 & 6.95 & 21.73 & 21.73 & 21.73 & 22.80 & 22.80 & 22.80 \\
\textbf{M5} & 12.39 & 37.39 & 37.39 & 37.39 & 28.77 & 28.77 & 28.77 & 8.87 & 22.91 & 22.91 & 22.91 & 24.39 & 24.39 & 24.39 \\
\hline
\multicolumn{15}{c}{\textit{Ten Nodes Network}} \\
\hline
\textbf{M1} & 3.96 & 41.72 & 34.13 & 28.09 & 24.86 & 20.67 & 29.79 & 2.92 & 24.18 & 20.58 & 17.33 & 21.61 & 17.97 & 25.95 \\
\textbf{M2} & 3.41 & 33.59 & 32.02 & 21.09 & 24.01 & 15.91 & 25.08 & 2.27 & 18.81 & 17.97 & 12.31 & 20.64 & 13.80 & 21.64 \\
\textbf{M3} & 4.66 & 34.17 & 32.66 & 21.21 & 22.36 & 16.58 & 23.36 & 3.11 & 18.83 & 18.03 & 11.38 & 18.81 & 14.68 & 20.20 \\
\textbf{M4} & 5.01 & 39.61 & 34.87 & 27.84 & 23.19 & 17.78 & 24.86 & 3.35 & 21.84 & 20.33 & 14.54 & 19.93 & 15.88 & 21.95 \\
\textbf{M5} & 6.48 & 41.16 & 36.45 & 28.98 & 24.78 & 19.03 & 27.91 & 4.47 & 22.63 & 20.37 & 16.22 & 20.18 & 17.28 & 22.74 \\
\hline
\end{longtable}

\vspace{-1.5em} 
\addtocounter{table}{-1} 

\begin{longtable}{|c||*{7}{c|}|*{7}{c|}}
\caption{\textit{(continued)}} \\ 
\hline
\multirow{2}{*}{\textbf{Model}} 
& \multicolumn{7}{c|}{\textbf{T = 500}} 
& \multicolumn{7}{c|}{\textbf{T = 1000}} \\
\cline{2-15}
& \textbf{EM1} & \textbf{EM2} & \textbf{EM3} & \textbf{EM4} & \textbf{EM5} & \textbf{EM6} & \textbf{EM7} 
& \textbf{EM1} & \textbf{EM2} & \textbf{EM3} & \textbf{EM4} & \textbf{EM5} & \textbf{EM6} & \textbf{EM7} \\
\hline
\endfirsthead

\multicolumn{15}{c}{\textit{(Table \ref{Sim:Table_no_miss_spectrum} continued - T=500, 1000)}} \\
\hline
\multirow{2}{*}{\textbf{Model}} 
& \multicolumn{7}{c|}{\textbf{T = 500}} 
& \multicolumn{7}{c|}{\textbf{T = 1000}} \\
\cline{2-15}
& \textbf{EM1} & \textbf{EM2} & \textbf{EM3} & \textbf{EM4} & \textbf{EM5} & \textbf{EM6} & \textbf{EM7} 
& \textbf{EM1} & \textbf{EM2} & \textbf{EM3} & \textbf{EM4} & \textbf{EM5} & \textbf{EM6} & \textbf{EM7} \\
\hline
\endhead

\multicolumn{15}{r}{\textit{(continued on next page)}} \\
\endfoot
\endlastfoot

\multicolumn{15}{c}{\textit{Five Nodes Network}} \\
\hline
\textbf{M1} & 3.88 & 13.09 & 14.72 & 13.37 & 21.47 & 19.02 & 21.38 & 2.82 & 9.18 & 11.72 & 9.95 & 15.46 & 13.13 & 17.44 \\
\textbf{M2} & 3.11 & 10.17 & 10.17 & 12.46 & 17.90 & 17.15 & 17.90 & 2.10 & 6.94 & 6.94 & 9.59 & 14.22 & 11.57 & 14.22 \\
\textbf{M3} & 3.94 & 10.72 & 10.72 & 10.72 & 18.64 & 18.64 & 18.64 & 2.89 & 7.53 & 7.53 & 7.53 & 14.72 & 14.72 & 14.72 \\
\textbf{M4} & 4.40 & 12.72 & 12.72 & 12.72 & 18.27 & 18.27 & 18.27 & 3.07 & 8.86 & 8.86 & 8.86 & 14.77 & 14.77 & 14.77 \\
\textbf{M5} & 5.80 & 13.74 & 13.74 & 13.74 & 19.28 & 19.28 & 19.28 & 3.99 & 9.43 & 9.43 & 9.43 & 15.23 & 15.23 & 15.23 \\
\hline
\multicolumn{15}{c}{\textit{Ten Nodes Network}} \\
\hline
\textbf{M1} & 1.98 & 13.67 & 13.40 & 11.87 & 17.83 & 14.96 & 20.93 & 1.37 & 9.30 & 10.81 & 9.95 & 15.46 & 13.13 & 17.44 \\
\textbf{M2} & 1.48 & 10.48 & 10.43 & 8.22 & 16.83 & 11.67 & 18.55 & 1.00 & 6.86 & 6.86 & 6.86 & 14.25 & 9.84 & 15.70 \\
\textbf{M3} & 1.98 & 10.95 & 10.38 & 8.51 & 15.69 & 12.50 & 16.11 & 1.16 & 7.39 & 7.39 & 6.71 & 13.87 & 10.48 & 13.59 \\
\textbf{M4} & 2.17 & 11.39 & 10.78 & 10.24 & 15.41 & 12.30 & 16.29 & 1.28 & 7.47 & 6.79 & 6.91 & 13.50 & 10.44 & 13.92 \\
\textbf{M5} & 2.68 & 11.90 & 10.84 & 10.35 & 15.86 & 12.44 & 17.22 & 1.54 & 7.93 & 7.22 & 7.38 & 13.99 & 11.06 & 14.71 \\
\hline
\end{longtable}

\section{Supplementary visualizations and details for the global bank network connectedness analysis in Section~\ref{sec:application}
}\label{app:D}

This appendix provides supplementary details for the global bank connectedness application introduced in Section~\ref{sec:intro} and analyzed in Section~\ref{sec:application}. Figure~\ref{fig:network2} offers a localized view of the Americas and European networks, while Figure~\ref{fig:Stg_net} illustrates the $r$-stage neighborhood structure. Additionally, the Corbit plot for GNAR order selection is displayed in Figure~\ref{fig:corbit}, with further analysis of disconnected nodes provided in Figure~\ref{fig:gnar_selected2}.

\renewcommand{\arraystretch}{1.3}
\setlength{\tabcolsep}{10pt}
\begin{table}[htbp]
\centering
\caption{Regional Distribution of Stationary Bank Codes. Details on the corresponding bank names are provided in the Online Appendix of \cite{Demirer2018}.}
\label{app:table_banks}
\begin{tabular}{|c|l|}
\hline
\textbf{Region} & \textbf{Banks} \\
\hline

\multirow{5}{*}{\textcolor{red}{Americas}} 
& \textcolor{red}{\texttt{jpm.us}, \texttt{bac.us}, \texttt{wfc.us}, \texttt{gs.us}} \\
& \textcolor{red}{\texttt{ms.us}, \texttt{td.ca}, \texttt{ry.ca}, \texttt{bns.ca}} \\
& \textcolor{red}{\texttt{bmo.ca}, \texttt{cm.ca}, \texttt{bk.us}, \texttt{bbdc4.br}} \\
& \textcolor{red}{\texttt{pnc.us}, \texttt{cof.us}, \texttt{stt.us}, \texttt{na.ca}} \\
& \textcolor{red}{\texttt{sti.us}, \texttt{fitb.us}, \texttt{axp.us}} \\
\hline

\multirow{5}{*}{\textcolor{darkblue}{Europe}} 
& \textcolor{darkblue}{\texttt{bnp.fr}, \texttt{aca.fr}, \texttt{gle.fr}, \texttt{san.es}} \\
& \textcolor{darkblue}{\texttt{ucg.it}, \texttt{ubsn.ch}, \texttt{csgn.ch}, \texttt{isp.it}} \\
& \textcolor{darkblue}{\texttt{bbva.es}, \texttt{stan.gb}, \texttt{dan.dk}, \texttt{kbc.be}} \\
& \textcolor{darkblue}{\texttt{dexb.be}, \texttt{bmps.it}, \texttt{sab.es}, \texttt{pop.es}} \\
& \textcolor{darkblue}{\texttt{bp.it}, \texttt{seba.se}, \texttt{dnb.no}} \\
\hline

\multirow{4}{*}{\textcolor{darkgreen}{Asia}} 
& \textcolor{darkgreen}{\texttt{mzh.jp}, \texttt{rsnh.jp}, \texttt{smtm.jp}, \texttt{wrfh.kr}} \\
& \textcolor{darkgreen}{\texttt{shf.kr}, \texttt{ibk.kr}, \texttt{ffg.jp}, \texttt{boy.jp}} \\
& \textcolor{darkgreen}{\texttt{cbb.jp}, \texttt{shzb.jp}, \texttt{bob.in}, \texttt{isctr.tr}} \\
& \textcolor{darkgreen}{\texttt{cmb1.cn}, \texttt{cmb2.cn}} \\
\hline

\multirow{1}{*}{\textcolor{darkviolet}{Australia}} 
& \textcolor{darkviolet}{\texttt{nab.au}, \texttt{cba.au}, \texttt{anz.au}, \texttt{wbc.au}} \\
\hline

\end{tabular}
\end{table}

\begin{figure}[!ht]
  \centering
  \includegraphics[width=0.5\linewidth]{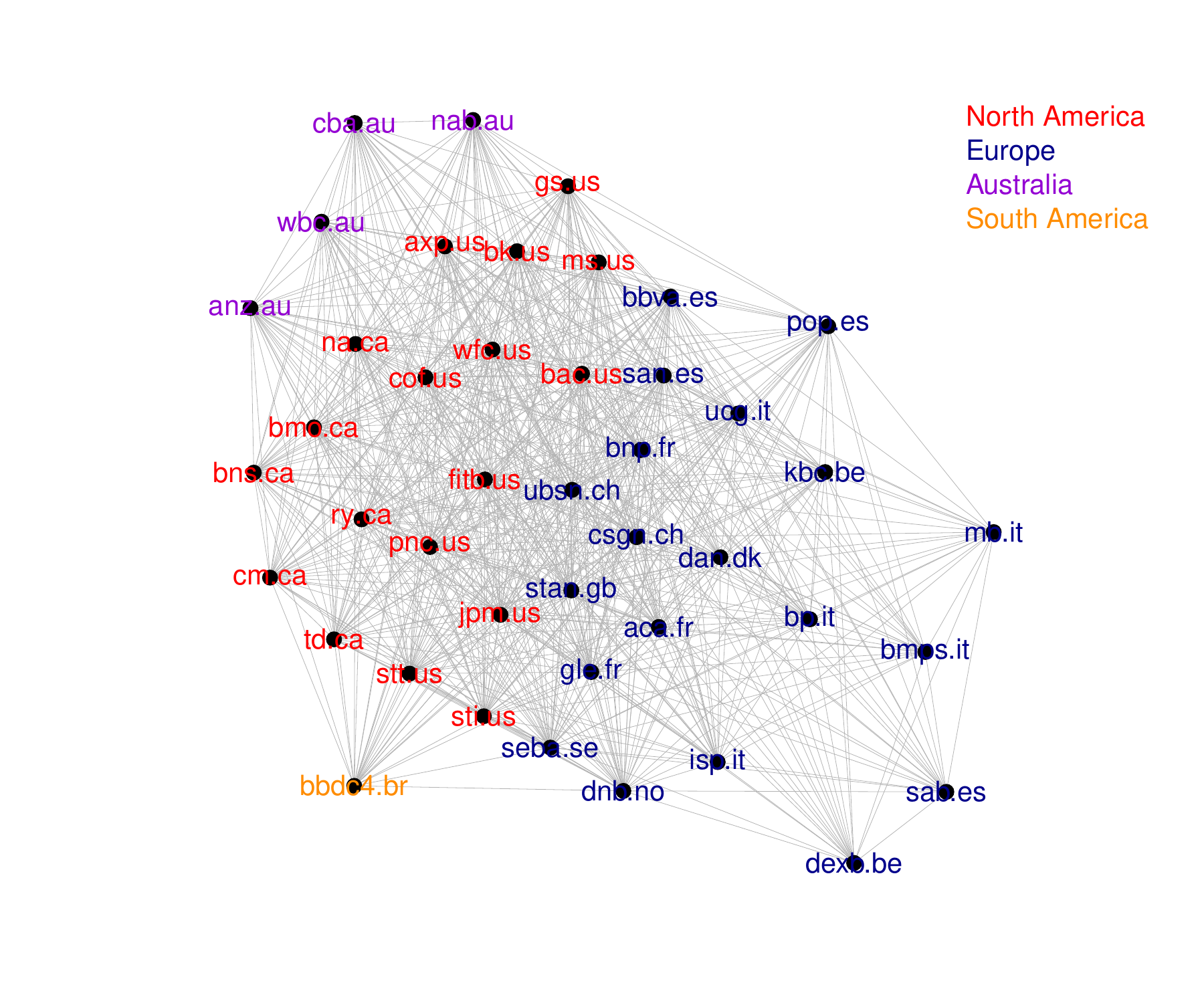}
  \caption{Zoomed-in adjacency network excluding Asian countries, estimated via GFEVD from log-volatility series using Lasso-VAR.}
  \label{fig:network2}
\end{figure}

\begin{figure}[!ht]
  \centering
  \includegraphics[width=0.5\linewidth]{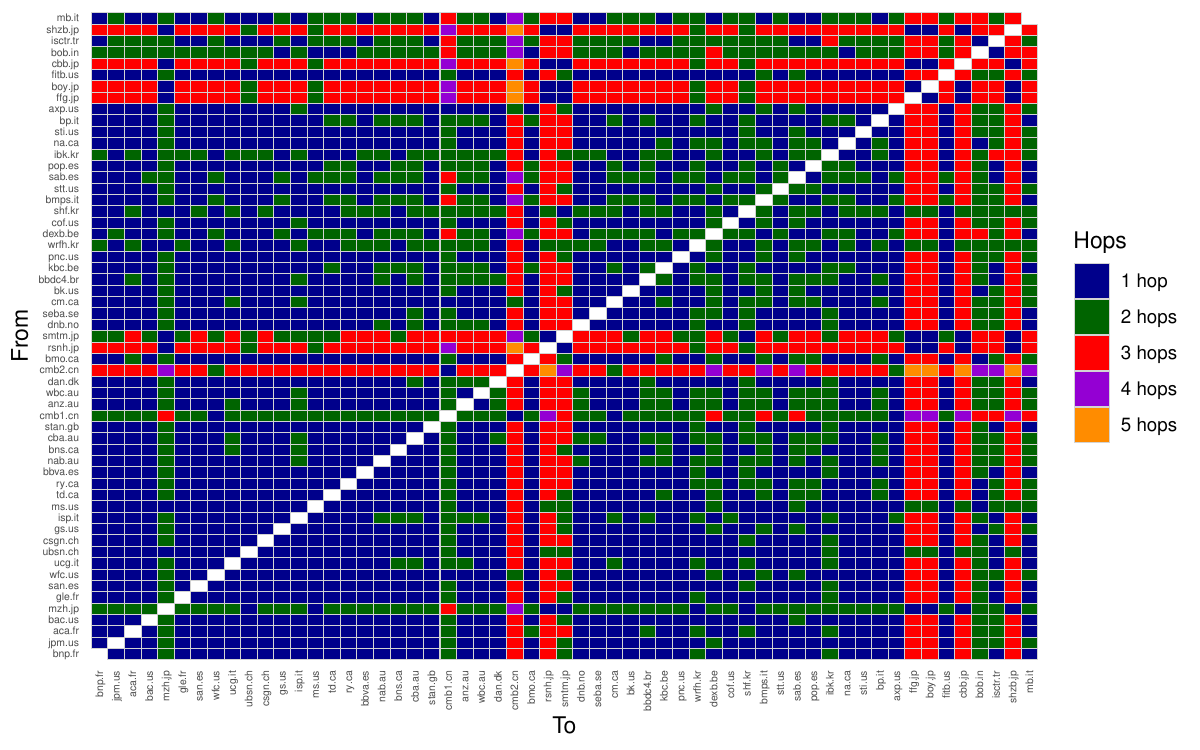}
  \caption{\( r \)-stage neighborhood structures in the global volatility network.}
  \label{fig:Stg_net}
\end{figure}

\begin{figure}[!ht]
  \centering
  \includegraphics[width=0.4\linewidth]{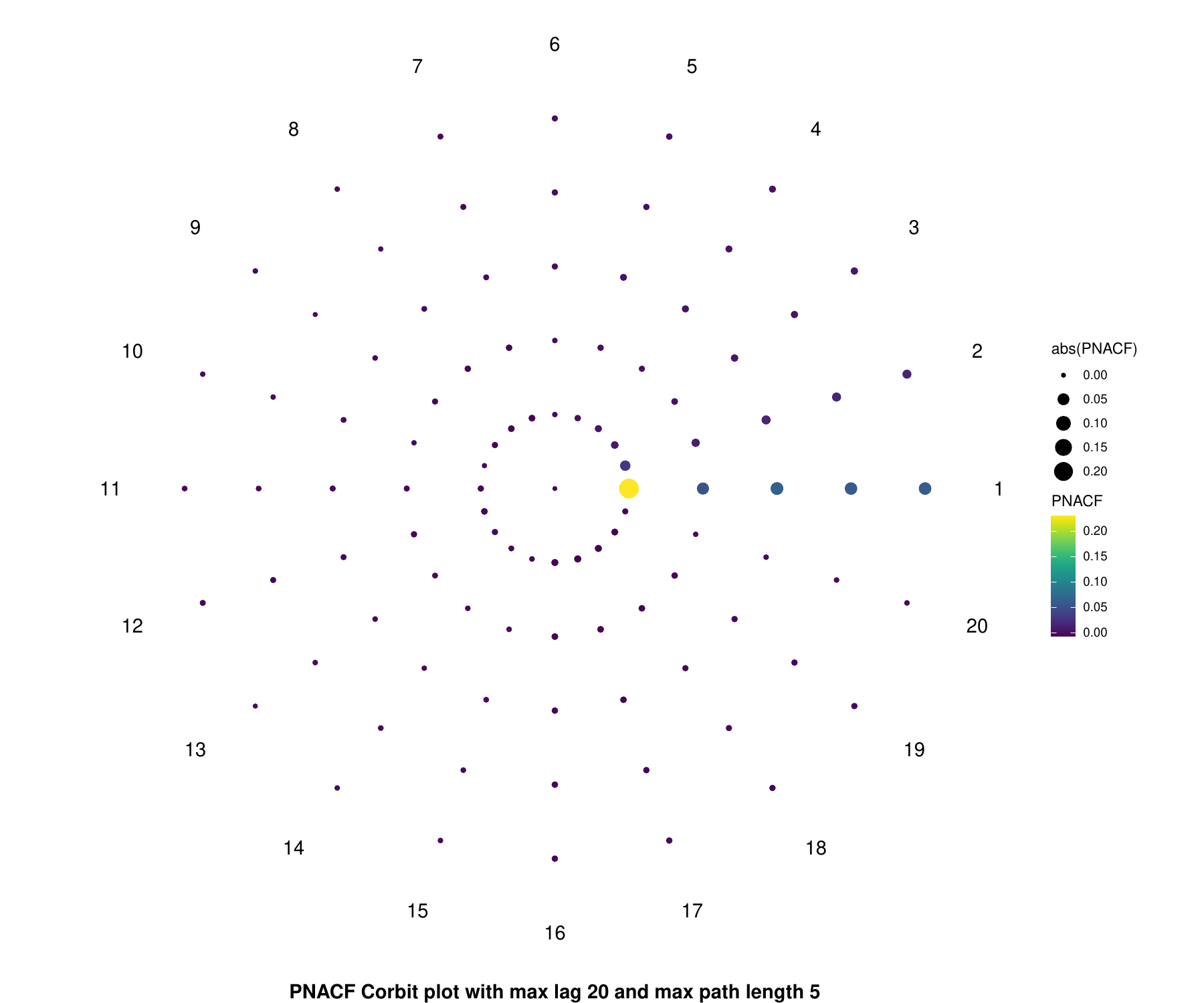}
  \caption{PNACF Corbit plot for the log-volatility series.}
  \label{fig:corbit}
\end{figure}

\begin{figure}[!ht]
    \centering
    \includegraphics[width=0.40\textwidth]{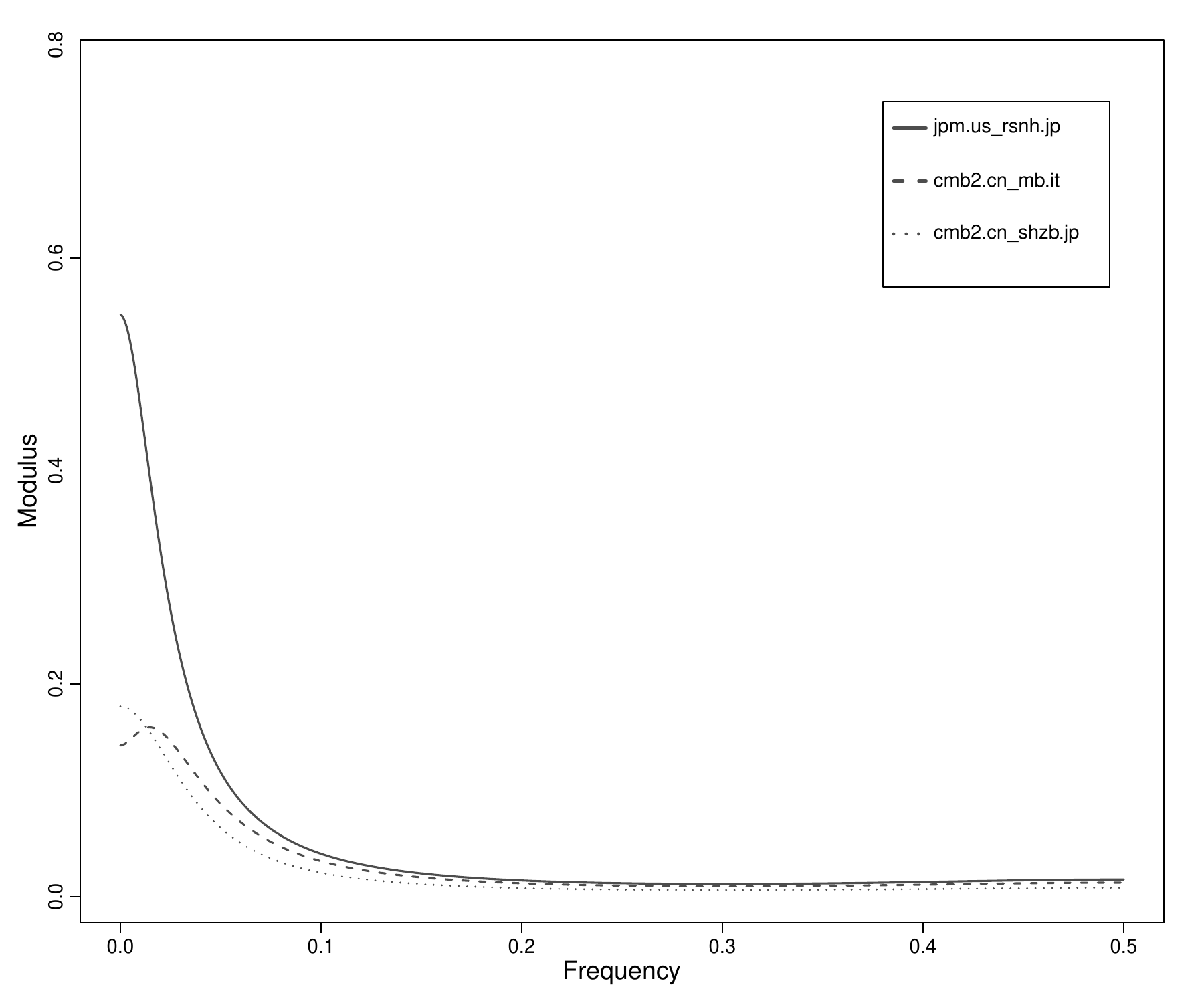} \quad
    \includegraphics[width=0.40\textwidth]{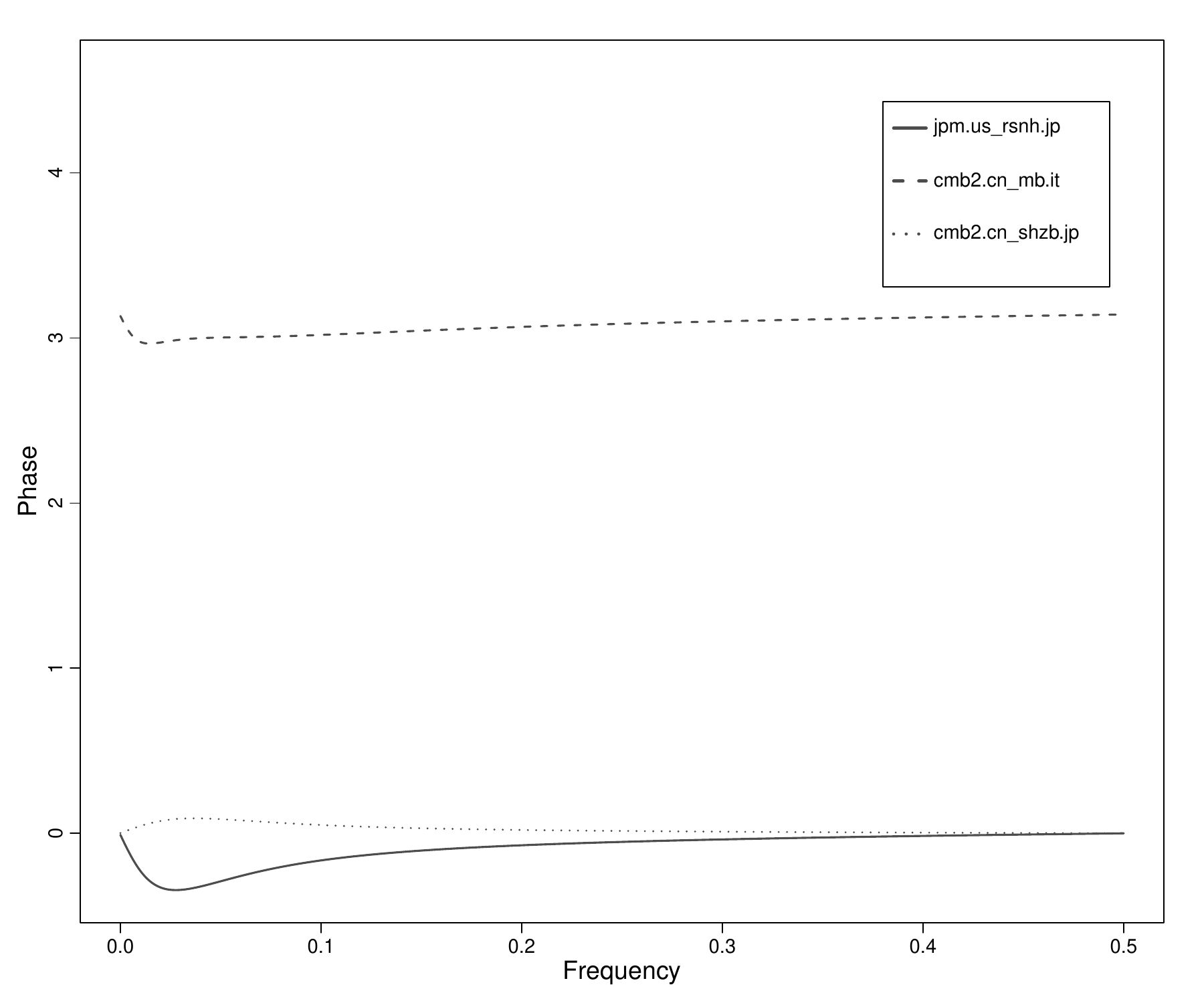} \\
    \vspace{0.5em}
    \includegraphics[width=0.40\textwidth]{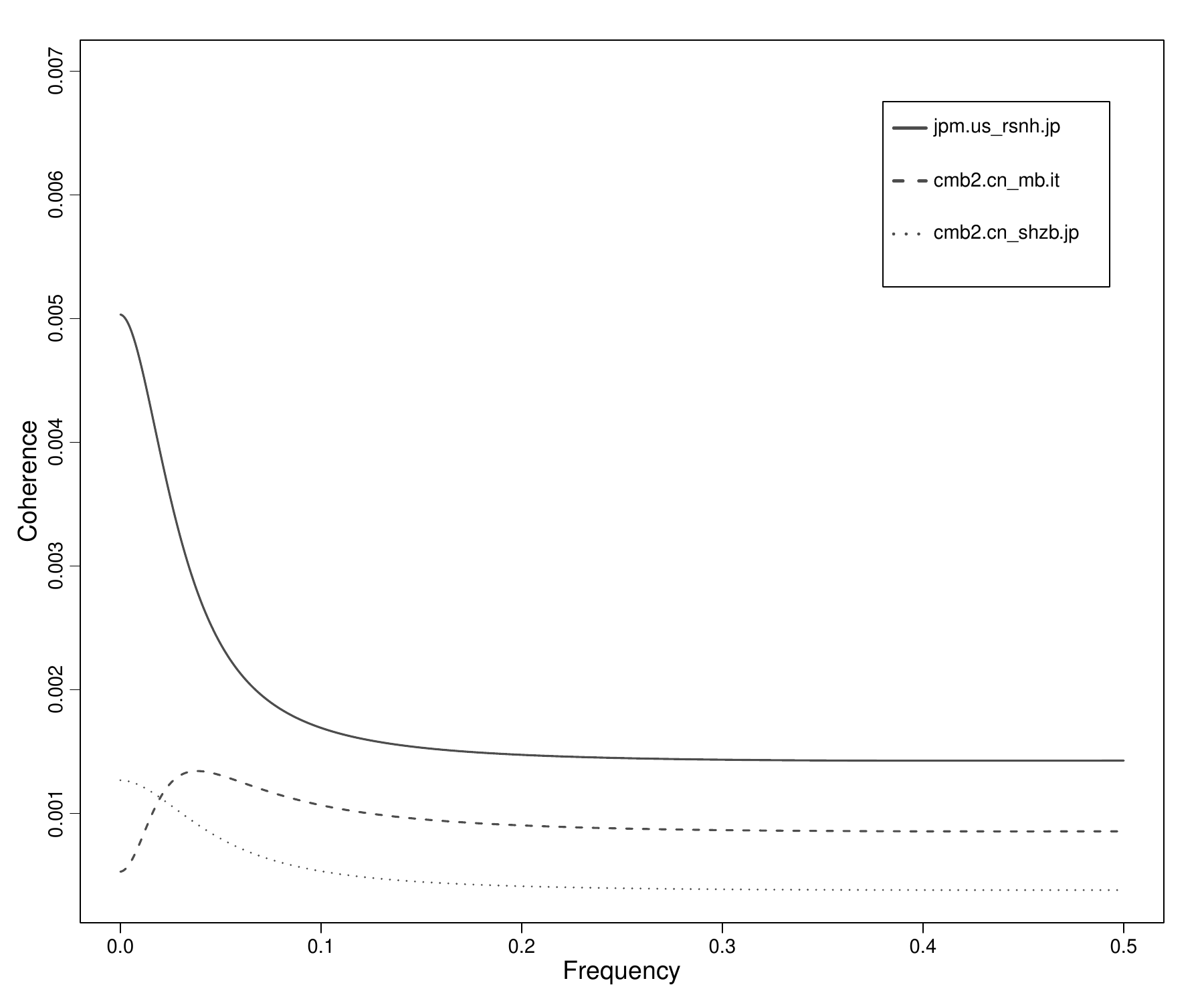} \quad
    \includegraphics[width=0.40\textwidth]{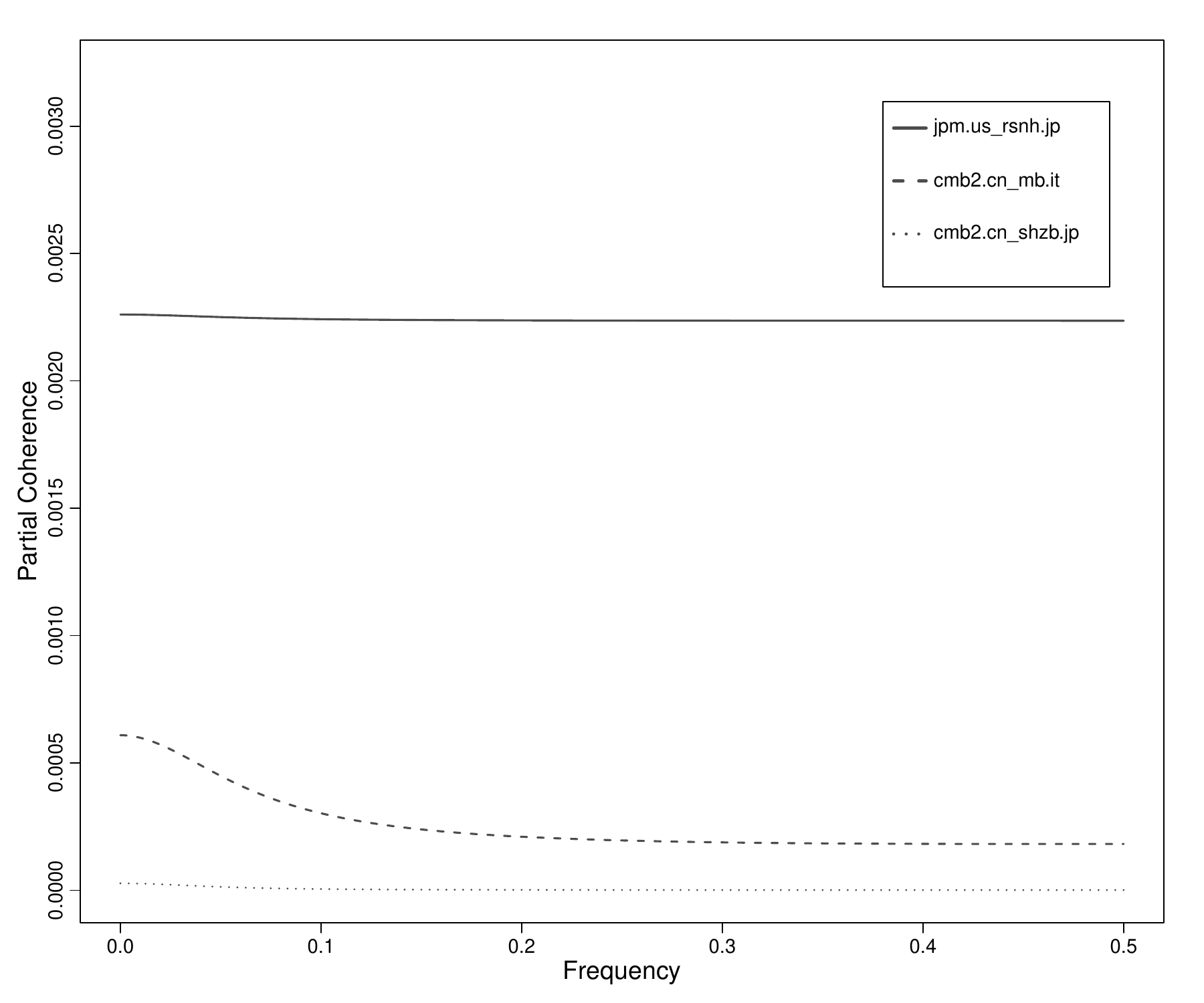}
    \caption{Spectral summaries for the three selected disconnected bank pairs: 
\textbf{JPMorgan Chase \& Co.}--\textbf{Resona Holdings} (\texttt{jpm.us}--\texttt{rsnh.jp}) ($r=3$), 
\textbf{China Merchants Bank}--\textbf{Mediobanca Banca di Credito Finanziario} (\texttt{cmb.cn}--\texttt{mb.it})($r=4$), 
and \textbf{China Merchants Bank}--\textbf{Shizuoka Bank} (\texttt{cmb.cn}--\texttt{shzb.jp})($r=5$). 
Top row: modulus \(|[\widehat{\mathbf{f}}(\cdotp,G)]_{ij}|\) and phase \(\arg([\widehat{\mathbf{f}}(\cdotp,G)]_{ij})\) of estimated cross-spectra. Bottom row: squared coherence \([\widehat{\boldsymbol{\rho}}(\cdotp,G)]_{ij}^2\) and partial coherence \([\widehat{\boldsymbol{\gamma}}(\cdotp,G)]_{ij}^2\). 
}
    \label{fig:gnar_selected2}
\end{figure}

\end{document}